\documentclass[russian, UKenglish]{lmcs}
\pdfoutput=1
\usepackage[utf8]{inputenc}

\usepackage{lastpage}
\lmcsdoi{21}{3}{10}
\lmcsheading{}{\pageref{LastPage}}{}{}%
{Aug.~22,~2024}{Jul.~28,~2025}{}


\usepackage{amsbsy,amssymb,amsmath,braket,pdfsync,babel}
\usepackage{xcolor}

\usepackage[only,llbracket,rrbracket]{stmaryrd}
\usepackage{xspace}
\usepackage{tikz}
\usetikzlibrary{positioning,shapes.geometric}

\newcommand{\cL}  {{\mathcal L}}
\newcommand{\cS}  {{\mathcal S}}
\newcommand{\PL}  {\ensuremath{\mathrm{PL}}}
\newcommand{\ML}  {\ensuremath{\mathrm{ML}}}
\newcommand{\MLK} {\ensuremath{\mathcal{S}_{\mathrm{K}}}\xspace}
\newcommand{\MLT} {\ensuremath{\mathcal{S}_{\mathrm{T}}}\xspace}
\newcommand{\MLSf}{\ensuremath{\mathcal{S}_{\mathrm{S5}}}\xspace}

\newcommand\N   {{\mathbb N}}
\newcommand\Bool{{\mathbb B}}
\newcommand\Var {{\mathcal V}}
\newcommand\Intr{{\mathcal I}}

\newcommand\dM   {\mathrm{DMor}}
\newcommand\extdM{\mathrm{DMor_\liff}}

\newcommand\true    {1}
\newcommand\false   {0}
\newcommand\lverum  {{\top}}
\newcommand\lfalsum {{\perp}}
\newcommand\liff    {\leftrightarrow}
\newcommand\limplies{\to}
\newcommand\nxt     {\lozenge}

\newcommand\formulabrackets[2][]{{#1\langle{#2}#1\rangle}}
\newcommand\formulasubst[2][]   {{#1({#2}#1)}}

\newcommand\rank        {\mathrm{rk}}
\newcommand\size[1]     {|#1|}
\newcommand\lnxtcount[1]{|#1|_\nxt}
\newcommand\sem[2]      {\llbracket #1;#2\rrbracket}

\newcommand\formulafont[1]           {\mathtt{#1}}
\newcommand{\formulatrue}            {\formulafont{tt}}
\newcommand{\formulafalse}           {\formulafont{ff}}
\newcommand{\formulanegation}        {\formulafont{neg}}
\newcommand{\formulapredisjunction}  {\formulafont{predis}}
\newcommand{\formuladisjunction}     {\formulafont{dis}}

\newcommand{\formulabiimplication}   {\formulafont{biimpl}}

\tikzset{
 formula/.style={
    baseline={([yshift={-\ht\strutbox}]current bounding box.north)},
    every node/.append style={inner sep=2pt},
    level distance=.65cm},
 structure/.style={
    baseline,
    -latex,
    every edge/.append style = {shorten <= 1pt}},
 world/.style={circle,draw=white,thick,inner sep=0,fill=black,
               minimum size=5pt},
 pin/.style={minimum size=2pt,inner sep=0pt},
 tree1/.style={isosceles triangle,draw,shape border rotate=90,
               isosceles triangle stretches=true, minimum height=12mm,
               minimum width=9mm,inner sep=0,rounded corners=3pt},
 tree2/.style={isosceles triangle,draw,shape border rotate=90,
               isosceles triangle stretches=true, minimum height=10mm,
               minimum width=10mm,inner sep=0,rounded corners=3pt,yshift=5pt},
 tree3/.style={isosceles triangle,draw,shape border rotate=90,
               isosceles triangle stretches=true, minimum height=11mm,
               minimum width=7mm,inner sep=0,rounded corners=3pt}
}

\title[Boolean basis, formula size, and number of modal
operators]{Boolean basis, formula size, and number of modal operators}

\author[Ch.~Berkholz]{Christoph Berkholz{\rsuper*}\lmcsorcid{0000-0002-3554-517X}}
\titlecomment{{\lsuper*}Christoph Berkholz was funded by the Deutsche Forschungsgemeinschaft (DFG, German Research Foundation) - project number 414325841.}
\author[D.~Kuske]{Dietrich Kuske}
\author[Ch.~Schwarz]{Christian Schwarz\lmcsorcid{0009-0002-6475-4779}}

\address{Technische Universit\"at Ilmenau, Germany}

\begin{document}

\begin{abstract}
  Is it possible to write significantly smaller formulae when using
  Boolean operators other than those of the De~Morgan basis (and, or,
  not, and the constants)? For propositional logic, a negative answer
  was given by Pratt: formulae over one set of operators can always be
  translated into an equivalent formula over any other complete set of
  operators with only polynomial increase in size.

  Surprisingly, for modal logic the picture is different: we show that
  elimination of bi-implication is only possible at the cost of an
  exponential number of occurrences of the modal operator~$\nxt$ and
  therefore of an exponential increase in formula size, i.e., the
  De~Morgan basis and its extension with bi-implication differ in
  succinctness. Moreover, we prove that any complete set of Boolean
  operators agrees in succinctness with the De~Morgan basis or with
  its extension with bi-implication.  More precisely, these results are
  shown for the modal logic $\mathrm{T}$ (and therefore for
  $\mathrm{K}$). We complement them by showing that the modal logic
  $\mathrm{S5}$ behaves as propositional logic: the choice of Boolean
  operators has no significant impact on the size of formulae.
\end{abstract}
\maketitle

\section{Introduction}

Many classical logics such as propositional logic, first-order and
second-order logic, temporal and modal logics incorporate a complete
set $G$ of Boolean operators in their definitions --- often the
De~Morgan basis consisting of the set of operators
$\dM=\{\land, \lor, \lnot, \lverum,\lfalsum\}$. But there are
certainly other options like $\{\to,\lfalsum\}$ or
$\{\operatorname{NAND}\}$. While for the expressivity it is clearly
irrelevant which complete operator set is used, this choice may have
an impact on how succinctly properties can be formulated. The main aim
of this paper is to understand the influence of the set of operators 
on the succinctness of formulae.

Suppose we extend the De~Morgan basis with an additional
operator.  If that additional operator defines a \emph{read-once
function}, i.e., it can be expressed in the De~Morgan basis in such
a way that every variable occurs at most once, then it can easily be
eliminated without blowing-up the formula too much. Thus, read-once
operators such as $x\to y \equiv \lnot x \lor y$ are really just
syntactic sugar.  For operators that are not read-once, such as
bi-implication $x\liff y$ or the ternary majority operator
$\mathrm{maj}(x,y,z)$, the situation is less clear, because mindlessly
replacing them with any equivalent De~Morgan formula may lead to an
exponential explosion of the formula size. So can it be that such
additional operators (or more generally, bases $G$ other than the
De~Morgan basis) actually allow us to write exponentially more succinct
formulae? For propositional logic, a negative answer was given by
Pratt \cite{Pra75}: for any two complete bases $F$ and $G$, first
balance the formula (that uses the basis $F$) so that it has
logarithmic depth (this step may introduce operators from the
De~Morgan basis) and
then replace all operators by any translation using the target
basis~$G$. This clearly leads to a linear increase in formula depth
and therefore only to a polynomial increase in formula size.

Balancing a formula is, however, not possible for logics that contain
quantifiers. For such logics it is still possible to efficiently
replace certain operators that are not read-once by
De~Morgan formulae. We show that if an operator
$\operatorname{op}(x_1,\ldots,x_k)$ is \emph{locally read-once}, that
is, has for every $i\in\{1,2,\dots,k\}$ an equivalent De~Morgan formula in which
$x_i$ appears only once, then it can be efficiently eliminated. An
example of an operator that is locally read-once, but not read-once,
is $\mathrm{maj}(x,y,z)$ (see Example~\ref{E:representations}). The
notion of a locally read-once operator and the algorithm of their
replacement can be generalized to any complete basis in place of the
De~Morgan basis. Thus, our first result says that the succinctness
does not differ much as long as operators of one basis are
locally read-once in the other and vice versa.

We then give a decidable characterisation of the operators that are locally
read-once in a given basis. This characterisation is based on the
notion of locally monotone operators: an operator is \emph{locally
  monotone} if fixing all but one argument defines a unary function
that is increasing or decreasing \emph{no matter how we fixed the
  remaining arguments} (e.g., bi-implication is not locally monotone,
but ternary majority is). Then an operator is locally read-once in a
given complete basis $G$ if, and only if, it is locally monotone or
some function from the basis $G$ is not locally monotone.

As a result, for any complete basis $G$, there are just two
possibilities: it allows to write formulae (up to a polynomial) as
succinct as the De~Morgan basis or as the De~Morgan basis extended
with the operator bi-implication $\liff$.\footnote{In the conference
version \cite{BerKS24} of this paper, this result was only shown for
extensions $G$ of the De~Morgan basis.} In other words, there are at
most two succinctness classes (recall that for propositional logic, there
is just one such class).

So far, the techniques and results hold for many classical logics (but
we spell them out in terms of modal logic). For modal logic, we
proceed by showing that there are indeed two different succinctness
classes. More precisely, we demonstrate that the use of operators that
are not locally monotone can be avoided, but only at the cost of an
exponential number of occurrences of the modal operator $\nxt$ and
therefore an exponential increase in formula size. Examples of such
useful operators are bi-implication $x\liff y$ and exclusive disjunction
$x \operatorname{XOR} y$.
In summary, there are exactly two succinctness classes that are
exponentially separated: one containing standard modal logic and the
other containing its extension with bi-implication.

Since this dichotomy is in contrast with propositional logic, where
only one succinctness class exists, we also investigate what happens
for fragments of modal logic defined by restrictions on the Kripke
structures. Here we obtain the same dichotomy for structures with a
reflexive accessibility relation. But upon considering equivalence
relations only, we can show that the two succinctness classes
collapse, as they do in propositional logic.  

\paragraph{Related work}
It seems that this paper is the first to
consider the influence of Boolean operators on the succinctness of
modal logics. Other aspects have been studied in detail.

Pratt \cite{Pra75} studied the effect of complete bases of binary
operators on the size of propositional formulae and proved in
particular that there are always polynomial translations.
Wilke \cite{Wil99} proved a succinctness
gap between two branching time temporal logics, Adler and Immerman
\cite{AdlI03} developed a game-theoretic method and used it to improve
Wilke's result and to show other succinctness gaps. The succinctness
of further temporal logics was considered, e.g., in
\cite{EteVW02,Mar03}.

Lutz et al.\ \cite{LutSW01,Lut06} study the succinctness and
complexity of several modal logics. French et al.\ \cite{FrHIK13}
consider multi-modal logic with an abbreviation that allows to express
``for all $i\in \Gamma$ and all $i$-successors, $\varphi$ holds''
where $\Gamma$ is some set of modalities. Using Adler-Immerman-games,
they prove (among other results in similar spirit) that this
abbreviation allows exponentially more succinct formulae than plain
multi-modal logic.

Grohe and Schweikardt \cite{GroS05} study the succinctness of
first-order logic with a bounded number of variables and, for that
purpose, develop extended syntax trees as an alternative view on
Adler-Immerman-games. These extended syntax trees were used by van
Ditmarsch et al.\ \cite{DitFHI14} to prove an exponential succinctness
gap between a logic of contingency (public announcement logic, resp.)
and modal logic.

Hella and Vilander \cite{HelV19} define a
formula size game (modifying the Adler-Immerman-game) and use it to
show that bisimulation invariant first-order logic is non-elementarily
more succinct than modal logic.

Immerman \cite{Imm81} defined separability games that characterise the
number of quantifiers needed to express a certain property in
first-order logic. These games were rediscovered and applied
\cite{FagLRV21,FagLVW22} and developed further (see \cite{CarFIKLSW24}
where further applications can be found). Hella and
Luosto~\cite{HelL24} defined alternative and equivalent
games. Vinall-Smeeth \cite{Vin24} studied the interplay of the number
of quantifiers and the number of variables needed to express certain
properties.

Since the modal operator $\nxt$ is a restricted form of
quantification, our result on the number of occurrences needed to
express certain properties is related to the works cited above.

\section{Definitions}

In this paper, $[n]=\{1,2,\dots,n\}$ for all $n\ge0$ and $0\in\N$.

\paragraph{Boolean functions I}
Let $\Bool=\{\false,\true\}$ denote the Boolean domain.
A \emph{Boolean function} is a function $f\colon \Bool^n\to\Bool$
for some $n\ge0$.
\allowdisplaybreaks
We consider the following functions that are usually called
disjunction, conjunction, implication, bi-implication, negation, falsum, 
verum, and majority:
\begin{align*}
  \lor     \colon &\Bool^2\to \Bool\colon (a,b)\mapsto \max(a,b)\\
  \land    \colon &\Bool^2\to \Bool\colon (a,b)\mapsto \min(a,b)\\
  \limplies\colon &\Bool^2\to \Bool\colon (a,b)\mapsto \max(1-a,b)\\
  \liff    \colon &\Bool^2\to \Bool\colon (a,b)\mapsto \max\bigl(\min(a,b),\min(1-a,1-b)\bigr)\\
  \lnot    \colon &\Bool^1\to \Bool\colon (a)  \mapsto \true-a\\
  \lfalsum \colon &\Bool^0\to \Bool\colon ()   \mapsto \false\\
  \lverum  \colon &\Bool^0\to \Bool\colon ()   \mapsto \true\\
  \mathrm{maj} \colon &\Bool^3\to \Bool\colon (a,b,c)   \mapsto
                        \begin{cases}
                          1 & \text{if }a+b+c\ge2\\
                          0 & \text{otherwise}
                        \end{cases}
\end{align*}

\paragraph{Propositional logic}

Let $\Var=\{p_i\mid i\ge0\}$ be a countably infinite set of
propositional variables. For a set $F$ of Boolean functions, let the
set of formulae of the propositional logic $\PL[F]$ be defined by
\begin{align*}
  \varphi &::=\ p \ \big|\ 
            f\formulabrackets{\underbrace{\varphi,\ldots,\varphi}_{k \text{ times}}} \,,
\end{align*}
where $p$ is some propositional variable and $f\in F$ is of arity $k$.
Note that we do not allow $\lfalsum$, $\lverum$, $\lor$, nor $\land$ 
in formulae unless they belong to $F$.
An example of a $\PL[\{\lnot,\lor\}]$-formula is
$\lnot\formulabrackets[\big]{\lor\formulabrackets{p_1,p_3}}$, which we
usually write $\lnot(p_1\lor p_3)$. We also write $\lfalsum$ for
$\lfalsum\formulabrackets{}$, but not in cases where we want to stress
the distinction between the nullary function $\lfalsum$ and the
formula $\lfalsum\formulabrackets{}$.

Let $\varphi\in\PL[F]$ and $y_1,\ldots,y_n$ be distinct propositional
variables. We write $\varphi\formulasubst{y_1,\ldots,y_n}$ to
emphasise that $\varphi$ uses at most the variables $y_1,\ldots,y_n$.
Let furthermore $\alpha_1,\ldots,\alpha_n\in\PL[F]$. Then
$\varphi\formulasubst{\alpha_1,\ldots,\alpha_n}$ denotes the
$\PL[F]$-formula obtained from $\varphi$ by substituting the
$\alpha_i$ for the $y_i$.

Let $\Intr$ be an \emph{interpretation} of the variables, i.e.,
a map $\Intr\colon\Var\to\Bool$. Inductively, we define the value 
$\Intr(\varphi)\in\Bool$ for formulae $\varphi\in\PL[F]$ via
\[
  \Intr\bigl(f\formulabrackets{\varphi_1,\dots,\varphi_k}\bigr)
  =f\bigl(\Intr(\varphi_1),\dots,\Intr(\varphi_k)\bigr)\,.
\]
We call two formulae $\varphi$ and $\psi$ \emph{equivalent} (denoted
$\varphi\equiv\psi$) if, for all interpretations $\Intr$, we have
$\Intr(\varphi)=\Intr(\psi)$. For instance, the formulae $x\lor y$ and
$y\lor x$ are equivalent, but also the formulae $x\lor\lnot x$ and
$\lverum$ (for any propositional variables $x$ and $y$).

For any sets of Boolean functions $F$ and $G$, we have
$\PL[F],\PL[G]\subseteq\PL[F\cup G]$, hence it makes sense to say that
formulae from $\PL[F]$ are equivalent to formulae from
$\PL[G]$.

\paragraph{Boolean functions II}

Let $\varphi\in\PL[F]$ be a formula that uses, at most, the variables 
from $\{y_1,\dots,y_n\}$. Then $\Intr(\varphi)$ depends on the values 
$\Intr(y_i)$ for $i\in[n]$, only (i.e., if $\Intr_1$ and $\Intr_2$ are
interpretations with $\Intr_1(y_i)=\Intr_2(y_i)$ for all $i\in[n]$, 
then $\Intr_1(\varphi)=\Intr_2(\varphi)$). Hence the formula $\varphi$
together with the sequence of variables $y_1,\dots,y_n$ defines a
function
\[
  \sem{\varphi}{y_1,\dots,y_n}\colon\Bool^n\to\Bool\colon
  (a_1,\dots,a_n)\mapsto\Intr(\varphi)\,,
\]
where $\Intr$ is any interpretation with $\Intr(y_i)=a_i$ for all
$i\in[n]$.

\begin{exa}\label{ex:a}
  Suppose $\varphi,\psi\in\PL[F]$ are equivalent and use (at most) the
  variables $\{y_1,\dots,y_n\}$. Then the functions
  $\sem{\varphi}{y_1,\dots,y_n}$ and $\sem{\psi}{y_1,\dots,y_n}$ are
  identical.

  Now consider the equivalent formulae
  $\varphi=\lverum\formulabrackets{}$ and $\psi=(x\lor\lnot x)$. By
  the above, $\sem{\varphi}{x}$ and $\sem{\psi}{x}$ are identical
  unary functions. But the nullary function $\sem{\varphi}{}$ is
  defined while $\sem{\psi}{}$ is undefined.
\end{exa}

Let $G$ be a set of Boolean functions. Then $G$ is said to be 
\emph{functionally complete} if for every Boolean function $f$ of
arity $k\geq 1$ there exists a $\PL[G]$-formula $\varphi(p_1,\dots,
p_k)$ such that $\sem{\varphi}{p_1,\dots,p_k}=f$. Standard examples
of functionally complete sets of Boolean functions are the
\emph{De~Morgan basis} $\dM=\{\lnot,\land,\lor,\lverum,\lfalsum\}$ and
the \emph{extended De~Morgan basis} $\extdM=\dM\cup\{\liff\}$.
In this paper, we use the following notion of completeness. A set $G$
of Boolean functions is \emph{complete} if for every $\PL[\dM]$-formula
there is an equivalent $\PL[G]$-formula (that may use more variables). 
Hence we consider $\dM$ as the canonical complete set of Boolean
functions and call $G$ complete if $\PL[G]$ is equally expressive
as $\PL[\dM]$. It is not difficult to see that the two notions of
completeness coincide: 
\begin{itemize}
  \item Assume that $G$ is functionally complete and let $\psi(p_1,
        \dots,p_k)\in\PL[\dM]$. Then $f=\sem{\psi}{p_1,\dots,p_{k+1}}$
        is a Boolean function of arity $k+1\geq 1$. Hence there is
        $\varphi(p_1,\dots,p_{k+1})\in\PL[G]$ such that
        \[
                  \varphi(p_1,\dots,p_{k+1})
           \equiv f\formulabrackets{p_1,\dots,p_{k+1}}
           \equiv \psi(p_1,\dots,p_k)\,,
        \]
        i.e., $G$ is complete.
  \item Conversely, assume that $G$ is complete and let $f$ be a
        Boolean function of arity $k\geq1$. Since $\dM$ is functionally
        complete and since $G$ is complete, there are formulae 
        $\psi(p_1,\dots,p_k)\in\PL[\dM]$ and $\varphi(p_1,\dots,p_\ell)
        \in\PL[G]$ with $\ell\geq k$ such that
        \[
                  f\formulabrackets{p_1,\dots,p_k}
           \equiv \psi(p_1,\dots,p_k)
           \equiv \varphi(p_1,\dots,p_\ell)\,.
        \]
        Let $\varphi'=\varphi(p_1,\dots,p_k,p_1,\dots,p_1)$. Since 
        $k\geq 1$, $\varphi'$ uses at most the variables $p_1,\dots,p_k$.
        Furthermore, $\varphi'\equiv\varphi$ and therefore $\sem{\varphi'}
        {p_1,\dots,p_k}=f$. This shows that $G$ is functionally complete.
\end{itemize}

\paragraph{Modal logic}

\noindent\textit{Syntax.}
For a set $F$ of Boolean functions, let the set of formulae of the
modal logic $\ML[F]$  be defined by 
\begin{align*}
  \varphi &::=\ p \ \big|\ 
            f\formulabrackets{\underbrace{\varphi,\ldots,\varphi}_{k \text{ times}}} \ \big|\ 
            \nxt\varphi\,,
\end{align*}
where $p$ is some propositional variable and $f\in F$ is of arity $k$.
The \emph{size} $\size{\varphi}$ of a formula $\varphi$ is the number
of nodes in its syntax tree.

\medskip

\noindent\textit{Semantics.}  Formulae are interpreted over \emph{pointed
Kripke structures}, i.e., over tuples $S=(W,R,V,\iota)$, consisting
of a set $W$ of possible worlds, a binary accessibility relation
$R\subseteq W\times W$, a valuation $V\colon \Var\to \mathcal P(W)$,
assigning to every propositional variable $p\in\Var$ the set of worlds
where $p$ is declared to be true, and an initial world $\iota\in W$.
The satisfaction relation $\models$ between a world $w$ of $S$ and an
$\ML[F]$-formula is defined inductively, where 
\begin{itemize}
  \item $S,w\models p$ if $w\in V(p)$,
  \item $S,w\models \nxt\varphi$ if $S,w'\models\varphi$
        for some $w'\in W$ with $(w,w')\in R$, and 
  \item $S,w\models f\formulabrackets{\alpha_1,\ldots,\alpha_k}$ if
        $f(b_1,\dots,b_k)=\true$ where, for all $i\in[k]$,  $b_i=\true$ iff
        $S,w\models\alpha_i$.
\end{itemize}
A pointed Kripke structure $S$ is a \emph{model} of $\varphi$
($S\models\varphi$) if $\varphi$ holds in its initial world, i.e.,
$S,\iota\models\varphi$.

Now let $\cS$ be some class of pointed Kripke structures.  A formula
$\varphi$ is \emph{satisfiable in $\cS$} if it has a model in
$\cS$ and $\varphi$ \emph{holds in $\cS$} if every
structure from $\cS$ is a model of $\varphi$. The
formula $\varphi$ \emph{entails} the formula $\psi$ \emph{in
$\cS$} (written $\varphi\models_{\cS}\psi$) if any
model of $\varphi$ from $\cS$ is also a model of~$\psi$;
$\varphi$ and $\psi$ are \emph{equivalent over $\cS$} (denoted
$\varphi\equiv_{\cS}\psi$) if $\varphi\models_{\cS}\psi$
and $\psi\models_{\cS}\varphi$.

\medskip

\noindent\textit{Classes of Kripke structures.} For different application
areas (i.e., interpretations of the operator $\nxt$), the following
classes of Kripke structures have attracted particular interest.  For
convenience, we define them as classes of \emph{pointed} Kripke
structures.
\begin{itemize}
\item The class $\MLK$ of all pointed Kripke structures. 
\item The class $\MLT$ of all pointed Kripke structures with reflexive 
      accessibility relation.
\item The class $\MLSf$ of all pointed Kripke structures where the 
  accessibility relation is an equivalence relation.
\end{itemize}

\noindent\textit{Succinctness and translations.}
Suppose $F$ and $G$ are two sets of Boolean functions and $G$ is
complete. Then the logic $\ML[G]$ is at least as expressive as the
logic $\ML[F]$, i.e., for any formula $\varphi$ from $\ML[F]$, there
exists an equivalent formula $\psi$ from $\ML[G]$. But what about the
size of $\psi$? Intuitively, the logic $\ML[G]$ is at least as
succinct as the logic $\ML[F]$ if the formula $\psi$ is ``not much
larger'' than the formula $\varphi$. This idea is formalized by the
following definition.

\begin{defi}[Translations]\label{def:succinctness}
  Let $F$ and $G$ be sets of Boolean functions, $\cS$ a class
  of pointed Kripke structures, and $\kappa\colon\N\to\N$ some
  function. Then $\ML[F]$ has \emph{$\kappa$-translations wrt.\
    $\cS$} in $\ML[G]$ if, for every formula
  $\varphi\in \ML[F]$, there exists a formula $\psi\in \ML[G]$ with
  $\varphi\equiv_{\cS}\psi$ and $|\psi|\le \kappa(|\varphi|)$.

  The logic $\ML[F]$ has \emph{polynomial translations wrt.\ $\cS$} in
  $\ML[G]$ if it has $\kappa$-translations wrt.~$\cS$ for some
  polynomial function $\kappa$. Finally, the logic $\ML[F]$ has
  sub-exponential translations wrt.\ $\cS$ if it has
  $\kappa$-translations wrt.~$\cS$ for some  function
  $\kappa$ with $\lim_{n\to\infty}\frac{\log\kappa(n)}{n}=0$.
\end{defi}

We consider two logics $\ML[F]$ and $\ML[G]$ equally succinct if the
former has polynomial translations in the latter and vice versa. It is
easily seen that this notion ``equally succinct'' is an equivalence
relation on the set of logics $\ML[F]$ for $F$ a complete set of
Boolean functions; we refer to the equivalence classes of this
relation as \emph{succinctness classes}. The following section will
demonstrate that there are at most two such succinctness classes,
namely those containing $\ML[\dM]$ and $\ML[\extdM]$,
respectively.\footnote{Recall that
  $\dM=\{\lnot,\lor,\land,\lverum,\lfalsum\}$ and
  $\extdM=\{\lnot,\lor,\land,\liff,\lverum,\lfalsum\}$ denote the
  (extended) De~Morgan basis.}

\section{"All" logics have at most two succinctness classes}
\label{sec:all}

The aim of this section is to show that, for any finite and complete set
of Boolean functions $F$, the logic $\ML[F]$ has the same succinctness
as the logic $\ML[\dM]$ or as the logic $\ML[\extdM]$
(Corollary~\ref{C-two-succinctness-classes}). 
Formally, one has
to be more precise since the relation ``equally succinct'' depends on
the class of pointed Kripke structures used to define the equivalence
of formulae. In this section, we consider the largest such class,
i.e., the class \MLK of all pointed Kripke structures.  For notational
convenience, we will regularly omit the explicit reference to the
class \MLK, e.g., ``equivalent'' means ``equivalent over \MLK'',
$\varphi\models\psi$ means $\varphi\models_{\MLK} \psi$, and
``$\kappa$-translations'' means ``$\kappa$-translations wrt.~\MLK''.

As to whether the logic $\ML[F]$ is in the succinctness class of
$\ML[\dM]$ or of $\ML[\extdM]$ depends on whether all functions from
$F$ are locally monotone:

\begin{defi}[Local monotonicity]
  Let $f\colon\Bool^k\to\Bool$ be a Boolean function. We say that
  $f$ is \emph{monotone in the $i$-th argument} if
  \begin{itemize}
  \item for all $\overline a\in\Bool^{i-1}$ and $\overline b\in\Bool^{k-i}$,
    $f(\overline a,\false,\overline b)\le f(\overline a,\true,\overline b)$ or
  \item for all $\overline a\in\Bool^{i-1}$ and $\overline b\in\Bool^{k-i}$,
    $f(\overline a,\false,\overline b)\ge f(\overline a,\true,\overline b)$.
  \end{itemize}
  The function $f$ is \emph{locally monotone} if it is monotone in every
  argument $i\in[k]$.
\end{defi}

Hence $f$ is monotone in the $i$-th argument, if, when changing the
$i$-th argument from $\false$ to~$\true$, while keeping the remaining
ones fixed, the value of $f$ uniformly increases or decreases (where,
in both cases, the value may also remain unchanged). By this definition,
conjunction, disjunction, negation, implication, as well as majority
are locally monotone functions, while bi-implication is not.

The main result of this section is a consequence of the following theorem.

\begin{thm}\label{T-main}
  Let $F$ and $G$ be finite sets of Boolean functions such that $G$ is
  complete.
  If all functions in $F$ are locally monotone or some function in $G$
  is not locally monotone, then $\ML[F]$ has polynomial translations in
  $\ML[G]$.
\end{thm}

The proof of the theorem can be found at the end of this section, in
the final Subsection~\ref{SS-proof-main-theorem}.
With Theorem~\ref{T-main}, we can already establish that the class of
logics $\ML[G]$ with $G$ finite and complete has at most two succinctness 
classes. Recall that we consider two logics equally succinct if the 
former has polynomial translations in the later and vice versa.

\begin{cor}\label{C-two-succinctness-classes}
  Let $G$ be some finite and complete set of Boolean functions. Then
  \begin{itemize}
  \item $\ML[G]$ and $\ML[\dM]$ are equally succinct wrt.~\MLK, or
  \item $\ML[G]$ and $\ML[\extdM]$ are equally succinct wrt.~\MLK .
  \end{itemize}
\end{cor}

\begin{proof}
  First recall that the De~Morgan basis $\dM$ and the extended De~Morgan 
  basis $\extdM=\dM\cup\{\liff\}$ are complete. The above theorem implies
  the following:
  \begin{itemize}
  \item Suppose that all functions from $G$ are locally monotone. Then
    $\ML[G]$ has polynomial translations in $\ML[\dM]$. But also all
    functions from $\dM$ are locally monotone; hence also $\ML[\dM]$
    has polynomial translations in $\ML[G]$. In other words, $\ML[G]$
    and $\ML[\dM]$ are equally succinct wrt.~\MLK in this case.
  \item Suppose that some function from $G$ is not locally monotone.
    Then the logic $\ML[\extdM]$ has polynomial translations in
    $\ML[G]$. The extended De~Morgan basis $\extdM$ contains the
    non-locally monotone function $\liff$; hence also $\ML[G]$ has
    polynomial translations in $\ML[\extdM]$. In other words, $\ML[G]$
    and $\ML[\extdM]$ are equally succinct wrt.~\MLK in this
    case.\qedhere
  \end{itemize}
\end{proof}

The remainder of this section is dedicated to the proof of 
Theorem~\ref{T-main}, which can be divided into the
following three steps.\label{Page-programme}

\paragraph{Step 1 (cf.\ Section~\ref{SS-step1})} We first show
that $\ML[F]$ has polynomial translations in $\ML[G]$ if the set $F$
of Boolean functions  admits ``$\PL[G]$-representations'' (to be
defined next). The idea is as follows:

Since $G$ is complete and therefore also functionally complete, there
is for every function $f\in F$ of arity $k\geq 1$ some formula $\omega
(p_1,\dots,p_k)\in\PL[G]$ such that $\omega(p_1,\dots,p_k)\equiv
f\formulabrackets{p_1,\dots,p_k}$. Consequently, in order to translate a
formula $\varphi\in\ML[F]$ into an equivalent formula $\psi\in\ML[G]$,
we only need to replace every sub-formula
$f\formulabrackets{\alpha_1,\dots,\alpha_k}$ in $\varphi$ by
$\omega\formulasubst{\beta_1,\dots,\beta_k}$ where $\beta_i$ is the
translation of $\alpha_i$ for $i\in[k]$.%
\footnote{There is a small issue here if $F$ contains a nullary
  function but $G$ does not. However, we ignore this for now by
  assuming that $G$ contains both $\lverum$ and $\lfalsum$ and
  handle the problem after Step~\ref{SS-step3} in a natural manner.}
In general, this translation leads to an exponential size increase.
But if, in the formula $\omega$, every variable $p_i$ appears only
once, we obtain a linear translation. The notion of representations 
is somewhat half-way between these two extremes.

\begin{defi}[Representations]\label{def:representation}
  Let $G$ be a set of Boolean functions, $f$ a Boolean function of
  arity $k>0$, and $i\in[k]$.

  A \emph{$\PL[G]$-representation of $(f,i)$} is a formula
  $\omega_i\formulasubst{p_1,\dots,p_k}\in\PL[G]$ that is equivalent
  to the formula $f\formulabrackets{p_1,\dots,p_k}\in\PL[\{f\}]$ and
  uses the variable $p_i$ at most once.

  A set $F$ of Boolean functions \emph{has $\PL[G]$-representations}
  if there are $\PL[G]$-representations for all $f\in F$ of arity $k>0$
  and all $i\in[k]$.
\end{defi}

\begin{exa}\label{E:representations}
  Consider the majority function $\mathrm{maj}(p_1,p_2,p_3)$ that is
  true iff at least two arguments are true. Then $(\mathrm{maj},1)$
  has the $\PL[\{\land,\lor\}]$-representation
  $\bigl(p_1 \land (p_2\lor p_3)\bigr) \lor (p_2\land p_3)$. Using the
  symmetry of $\mathrm{maj}$, it follows that $\{\mathrm{maj}\}$ has
  $\PL[\{\land,\lor\}]$-representations.

  Next, consider bi-implication $\liff$. Recall that a
  $\PL[\{\lnot,\land,\lor\}]$-formula $\psi(p_1,p_2)$ that contains $p_1$
  only under an even number of negations is monotonely increasing in $p_1$;
  analogously, if $p_1$ occurs only under an odd number of negations, then 
  $\psi$ is monotonely decreasing in $p_1$.
  Aiming at a contradiction, assume that $(\liff,1)$ had a 
  $\PL[\{\lnot,\land,\lor\}]$-representation $\omega(p_1,p_2)$ that mentions
  $p_1$ only once (say, under an even number of negations). Then, by the
  previous observation, flipping $p_1$ from $0$ to $1$ does not
  decrease the truth value of the formula, hence
  \[
    1 = (0\liff 0) = \sem{\omega}{p_1,p_2}(0,0) \le
    \sem{\omega}{p_1,p_2}(1,0) = (1\liff 0) = 0\,,
  \]
  a contradiction. Thus $\{\liff\}$ does not have
  $\PL[\{\lnot,\land,\lor\}]$-representations.
\end{exa}

\paragraph{Step 2 (cf.\ Section~\ref{SS-step2})}  Having established
in step 1 that representations yield polynomial translations, we next
show that any Boolean function has $\PL[\extdM]$-representations and
that any locally monotone function even has $\PL[\dM]$-representations.

\paragraph{Step 3 (cf.\ Section~\ref{SS-step3})} In this final step, we 
construct $\PL[G]$-representations of all functions from $\dM$ and 
$\extdM$ provided $G$ is complete (and contains some non-locally monotone
function for the case $\extdM$).

\paragraph{Summary (cf.~Section~\ref{SS-proof-main-theorem})} Now
suppose that $G$ is complete and that all functions from $F$ are locally
monotone or some function from $G$ is not locally monotone. In the
first case, $F$ and $\dM$ have $\PL[\dM]$- and
$\PL[G]$-representations, respectively (step 2 and 3). Consequently,
$\ML[F]$ and $\ML[\dM]$ have polynomial translations in $\ML[\dM]$ and
$\ML[G]$, respectively (step 1). By transitivity, $\ML[F]$ has
polynomial translations in $\ML[G]$.

In the second case (i.e., $G$ contains some function that is not locally
monotone), we can argue similarly but using $\extdM$ in place of $\dM$.
 
\subsection{Representations yield polynomial translations}
\label{SS-step1}

Let $F$ and $G$ be two sets of Boolean functions such that $F$ has
$\PL[G]$-representations. We will construct from a formula in $\ML[F]$
an equivalent formula in $\ML[G]$ of polynomial size. Since this will
be done inductively, we will have to deal with formulae from
$\ML[F\cup G]$ and the task then is better described as elimination of
functions $f\in F\setminus G$ from formulae in $\ML[F\cup G]$.

Before we present the details of our construction, we briefly
demonstrate the main idea behind the proof (for $F=\{\land,\lor,f\}$
and $G=\{\land,\lor,\lnot\}$, i.e., we aim to eliminate some Boolean
function $f$).  The main results (Lemma~\ref{thm:@lem} and
Proposition~\ref{thm:mon-poly-succinct@prop}) will appear at the end
of the section.
 
Assume that $f$ is of arity $2$ and consider the two formulae
\begin{align*}
  \varphi &=
            \bigl(
             p_1\lor
             f\formulabrackets[\big]{\ f\formulabrackets{p_2,p_3},\ p_4\land p_4\ }
            \bigr)\lor
              f\formulabrackets{p_6,p_7\land p_8}\ 
  \text{ and } \\
  \psi    &= f\formulabrackets[\big]{\ 
              p_1\lor  f\formulabrackets[\big]{\ f\formulabrackets{p_2,p_3},\ p_4\land p_5\ },\ 
              f\formulabrackets{p_6,p_7\land p_8}\ 
            }\,,
\end{align*}
whose syntax trees are depicted in Fig.~\ref{fig:syntax-tree} (they
only differ in the root node).
\begin{figure}
  \centering
  \begin{tikzpicture}[
         formula,
         level/.style={sibling distance=2.5cm},
         level 2/.style={sibling distance=1cm},
         level 3/.style={sibling distance=1.5cm},
         level 4/.style={sibling distance=.5cm}]
    \node {$\lor$}
     child { node{$\lor$} 
       child { node{$p_1$} }
       child { node{$f$} 
         child { node{$f$} 
           child { node{$p_2$} }
           child { node{$p_3$} } }
         child { node{$\land$}
           child { node{$p_4$} }
           child { node{$p_5$} } }
         }
     }
     child { node{$f$}
       child { node{$p_6$} }
       child { node{$\land$}
         child { node{$p_7$} }
         child { node{$p_8$} } }
     };
  \end{tikzpicture}
  \hspace*{2cm}
  \begin{tikzpicture}[
         formula,
         level/.style={sibling distance=2.5cm},
         level 2/.style={sibling distance=1cm},
         level 3/.style={sibling distance=1.5cm},
         level 4/.style={sibling distance=.5cm}]
    \node {$f$}
     child { node{$\lor$} 
       child { node{$p_1$} }
       child { node{$f$} 
         child { node{$f$} 
           child { node{$p_2$} }
           child { node{$p_3$} } }
         child { node{$\land$}
           child { node{$p_4$} }
           child { node{$p_5$} } }
         }
     }
     child { node{$f$}
       child { node{$p_6$} }
       child { node{$\land$}
         child { node{$p_7$} }
         child { node{$p_8$} } }
     };
  \end{tikzpicture}
  \caption{Syntax trees of the formulae
    $\varphi =
            \bigl(
             p_1\lor
             f\formulabrackets{f\formulabrackets{p_2,p_3},p_4\land p_5}
            \bigr)\lor
              f\formulabrackets{p_6,p_7\land p_8}$ and
  $\psi    = f\formulabrackets{
              p_1\lor  f\formulabrackets{f\formulabrackets{p_2,p_3},p_4\land p_5}, 
              f\formulabrackets{p_6,p_7\land p_8}
            }$.
  }
  \label{fig:syntax-tree}
\end{figure}

A distinguishing property of the left tree is that there is no $f$-node
with left immediate successor $\ell$ and right immediate successor $r$
such that the sub-trees with roots $\ell$ and $r$ both contain an 
$f$-node. Assuming that $f$
has $\PL[G]$-representations, there exist Boolean combinations
$\omega_1(x,y)\equiv\omega_2(x,y)\equiv f\formulabrackets{x,y}$ of the
variables $x$ and $y$, such that $x$ occurs only once in
$\omega_1(x,y)$ and $y$ only once in $\omega_2(x,y)$.  Proceeding
bottom-up, we now replace each $f$-node
$f\formulabrackets{\alpha,\beta}$ in the syntax tree of $\varphi$ by
either $\omega_1(\alpha,\beta)$ or $\omega_2(\alpha,\beta)$, depending
on whether we have previously modified the left or the right sub-tree
(regarding, e.g., $f\formulabrackets{p_2,p_3}$, we are free to choose
between $\omega_1$ and $\omega_2$).  Note that, although $\omega_1$
and $\omega_2$ may duplicate some parts of $\varphi$, our choice
ensures that we never duplicate such parts whose size has already
changed. Consequently, this procedure results in a linear increase
$\ell\cdot|\varphi|$ in the size of $\varphi$, where the coefficient
$\ell$ essentially depends on how often $y$ occurs in $\omega_1(x,y)$
and how often $x$ occurs in $\omega_2(x,y)$. The resulting formula
$\varphi'$ belongs to $\ML[G]$ and is equivalent to $\varphi$. Hence
$\ML[F\cup G]$-formulae where each $f$-node in the syntax tree
has at most one immediate
successor $s$ such that the sub-tree with root $s$ contains an $f$-node
have $\ML[G]$-translations of linear size.

The formula $\psi=f\formulabrackets{\alpha,\beta}$ on the other hand
does not have this property, but the two sub-formulae 
$\alpha=p_1\lor f\formulabrackets{f\formulabrackets{p_2,p_3},p_4\land p_5}$
and $\beta=f\formulabrackets{p_6,p_7\land p_8}$ do.
Hence we can apply the above transformation to them
separately, yielding equivalent $\ML[G]$-formulae $\alpha'$ and
$\beta'$ whose size increases at most by a factor of $\ell$.  Then
$\psi' = f\formulabrackets{\alpha',\beta'}\equiv\psi$ is of size
$|\psi'|\le \ell \cdot|\psi|$. Note that this step reduces the total
number of $f$-vertices. In our example, $\alpha',\beta'\in\ML[G]$ and
there remains only the single $f$-node which is the root of the
syntax tree of $\psi'=f\formulabrackets{\alpha',\beta'}$. Applying the
step again yields an equivalent $\ML[G]$-formula $\psi''$ of
size $|\psi''|\leq\ell\cdot |\psi'|\leq\ell^2\cdot|\psi|$.

Consider now an arbitrary formula $\chi\in\PL[G\cup F]$ and let $D$
denote the number of steps that are required until all $f$-nodes 
have been removed from $\chi$, i.e., until a $\PL[G]$-formula is obtained. 
Then the resulting formula has size at most $\ell^D\cdot|\chi|$. In this 
section, we will show that $D$ is at most logarithmic in the size of $\chi$.
Both results together give a polynomial bound on the size of the 
equivalent formula from $\ML[G]$, which establishes the main part of
the succinctness result.

We now formalize this idea and prove the result rigorously. Let $F$
and $G$ be disjoint sets of Boolean functions (we assume that $F$ and
$G$ are disjoint for notational convenience). Let
$N_{F,G} \subseteq \ML[F\cup G]$ be defined by
  \begin{align*}
    \varphi &::=\ \psi \ \big|\ 
                  f\formulabrackets{\psi,\ldots,\psi,\varphi,\psi,\ldots,\psi} \ \big|\ 
                  g\formulabrackets{\varphi,\ldots\varphi} \ \big|\ 
                  \nxt\varphi
  \end{align*}
for $\psi\in\ML[G]$, $f\in F$, and $g\in G$.
Note that a formula $\varphi$ belongs to $N_{F,G}$ if, and only if, 
for any sub-formula $f\formulabrackets{\alpha_1,\dots,\alpha_k}$ of 
$\varphi$ with $f\in F$, at most one $\alpha_i$ contains some function
from $F$. There is no such restriction on the sub-formulae 
$g\formulabrackets{\alpha_1,\ldots,\alpha_k}$ of $\varphi$ where $g\in G$.
Note that a formula in $N_{F,G}\setminus\ML[G]$ contains some sub-formula
of the form $f\formulabrackets{\alpha_1,\ldots,\alpha_k}$ (and all but at
most one $\alpha_i$ belong to $\ML[G]$).

\begin{defi}[Derivative]
  Let $F$ and $G$ be disjoint sets of Boolean functions and let
  $\varphi=\varphi\formulasubst{p_1,\ldots,p_m}\in \ML[F\cup G]$.
  The \emph{$F$-derivative $d_F(\varphi)$ of $\varphi$} is the smallest
  $\ML[F\cup G]$-formula $\gamma\formulasubst{p_1,\ldots,p_m,q_1,\ldots,q_n}$
  (up to renaming of the variables $q_1,\ldots,q_n$), such that
  \begin{itemize}
    \item $q_i$ occurs exactly once in 
          $\gamma\formulasubst{p_1,\ldots,p_m,q_1,\ldots q_n}$
          for all $i\in[n]$ and
    \item there exist $\alpha_1,\ldots,\alpha_n\in N_{F,G}\setminus\ML[G]$
          such that $\varphi$ and 
          $\gamma\formulasubst{p_1,\ldots,p_m,\alpha_1,\ldots,\alpha_n}$ 
          are identical.
  \end{itemize}
\end{defi}
Intuitively, $\gamma\formulasubst{p_1,\ldots,p_m,q_1,\ldots,q_n}$ is
obtained from $\varphi\formulasubst{p_1,\ldots,p_m}$ by simultaneously 
replacing all ``maximal $(N_{F,G}\setminus \ML[G])$-formulae'' by
distinct fresh variables $q_1,\ldots,q_n$ (where multiple occurrences 
of the same formula are replaced by different variables).
An example is depicted in Fig.~\ref{fig:derivative-example} (with
$F=\{f,f',f''\}$ and $G=\{\lnot,\land,\lor\}$).
\begin{figure}
  \centering
  \begin{tikzpicture}[
         formula,
         level/.style={sibling distance=2cm},
         level 2/.style={sibling distance=1.1cm},
         level 3/.style={sibling distance=.5cm}]
    \node (phi) {$f$}
     child { node{$f$} 
       child { node{$f'$} 
         child { node{$p_1$} }
         child { node{$p_2$} }
       }
       child { node{$f'$} 
         child { node{$p_1$} }
         child { node{$p_2$} }
       }
     }
     child { node{$f''$}
       child { node{$p_2$} }
       child { node{$f$}
         child { node{$p_1$} }
         child { node{$p_1$} }
       }
     };

    \node[right=4cm of phi] (dphi) {$f$}
     child { node{$f$}
       child { node{$q_1$} }
       child { node{$q_2$} }
     }
     child { node{$q_3$} };

    \node[right=2.5cm of dphi] (ddphi) {$q_1'$};
  \end{tikzpicture}
  \caption{Syntax trees of
    $\varphi=f\formulabrackets[\bigl]{
      f\formulabrackets[\bigl]{
        f'\formulabrackets{p_1,p_2},
        f'\formulabrackets{p_1,p_2}},\,
      f''\formulabrackets{p_2,f\formulabrackets{p_1,p_1}}}$,
    $d_{F}(\varphi)=f\formulabrackets[\bigl]{f\formulabrackets{q_1,q_2},q_3}$, and
    $d_F^2(\varphi)=q_1'$. 
  }
  \label{fig:derivative-example}
\end{figure}

Let $\varphi\in\ML[F\cup G]$ be a formula that contains some function
from $F$.  Then $d_F(\varphi)$ contains fewer occurrences of functions
from $F$ than $\varphi$. Hence there exists a smallest integer
$r\geq 0$ for which the $r$-th derivative
$d_F^r(\varphi)=d_F(d_F(\cdots d_F(\varphi) \cdots))$ is an
$\ML[G]$-formula, where $d_F^0(\varphi)=\varphi$.
\begin{defi}[Rank]
  Let $\varphi\in \ML[F\cup G]$. The \emph{$F$-rank $\rank_F(\varphi)$
  of $\varphi$} is the smallest integer $r\geq 0$ for which
  $d_F^r(\varphi)\in \ML[G]$.
\end{defi}

We first show that a formula with high $F$-rank must also be large.
\begin{lem}\label{thm:bound-rank@lem}
  Let $F$ and $G$ be disjoint sets of Boolean functions and let
  $\varphi\in \ML[F\cup G]$. Then $|\varphi|\geq 2^{\rank_F(\varphi)}-1$.
\end{lem}

\begin{proof}
  Throughout this proof, we will refer to the $F$-rank of a formula simply
  as rank.
  
  We prove the stronger claim that a formula $\varphi$ of positive rank 
  has at least $2^{\rank_F(\varphi)-1}$ sub-formulae (counting multiplicity)
  of the form $f\formulabrackets{\beta_1,\dots,\beta_k}$ with $f\in F$ and
  $\beta_1,\dots,\beta_k\in\ML[G]$. In the following, we call such formulae
  \emph{$F$-leaves}.
  Since counting the $F$-leaves in all derivatives of $\varphi$ results 
  in a lower bound on the total number of functions from $F$ in $\varphi$, 
  it follows that $\varphi$ contains at least
  $2^0+2^1  +\ldots+2^{\rank_F(\varphi)-1}=2^{\rank_F(\varphi)}-1$ 
  functions from $F$.
  
  It now remains to prove the bound on the number of $F$-leaves.
  Recall that we consider formulae of rank at least one.
  If $\rank_F(\varphi)=1$, $\varphi$ contains at least one function from 
  $F$ and hence has at least one $F$-leaf.
  Now, assume that $\varphi=\varphi\formulasubst{\overline p}$ is of rank
  at least two, where $\overline p = (p_1,\ldots,p_m)$.
  Let $\gamma\formulasubst{\overline p,q_1,\ldots,q_n}$ be the
  derivative of $\varphi$ and let $\alpha_1,\ldots,\alpha_n\in 
  N_{F,G}\setminus\ML[G]$ with
  $\varphi=\gamma\formulasubst{\overline p,\alpha_1,\ldots,\alpha_n}$. 
  Then, for every $F$-leaf $f\formulabrackets{\beta_1,\ldots,\beta_k}$ of 
  $\gamma\formulasubst{\overline p,q_1,\ldots,q_n}$, there exist indices 
  $1\leq i<j\leq k$ such that $\beta_i,\beta_j\in\{q_1,\ldots,q_n\}$. 
  By induction hypothesis, $\gamma\formulasubst{\overline p,q_1,\ldots,q_n}$
  has at least $2^{\rank_F(\varphi)-2}$ $F$-leaves, each of which contains
  at least two of the fresh variables $\{q_1,\ldots,q_n\}$. Since each 
  $\alpha_i$ contains at least one function from $F$, 
  $\varphi=\gamma\formulasubst{\overline p,\alpha_1,\ldots,\alpha_n}$ has 
  at least twice the number of $F$-leaves compared to 
  $\gamma\formulasubst{\overline p,q_1,\ldots,q_n}$, i.e., at least
  $2^{\rank_F(\varphi)-1}$ $F$-leaves.
\end{proof}

We can now turn to the main ingredient for our succinctness result.

\begin{lem}\label{thm:@lem}
  Let $F$ and $G$ be disjoint sets of Boolean functions,
  $\lverum,\lfalsum\in G$ (such that no function from $F$ is nullary),
  and, for $f\in F$ and $i\in[\mathrm{ar}(f)]$, let
  $\omega_{f,i}\in\PL[G]$ be a $\PL[G]$-representation of
  $(f,i)$. Let, furthermore, $\kappa\colon\N\to\N$ be a monotone
  function such that $1\leq\kappa(0)$ and
  $|\omega_{f,i}|\le \kappa(\mathrm{ar}(f))$ for any $f\in F$ and
  $i\in[\mathrm{ar}(f)]$. Finally, let
  $\kappa'\colon\N\to\N\colon n\mapsto \kappa(n)^{1+\log_2 n}\cdot n$.

  Then $\ML[F\cup G]$ has  $\kappa'$-translations in $\ML[G]$.
\end{lem}

\begin{proof}
  For $\varphi\in\ML[F\cup G]$, let $K_\varphi$ denote the maximal arity
  of any $f\in F$ occurring in $\varphi$ (or $0$ if $\varphi\in \ML[G]$).
  
  We prove the following claim: every formula
  $\varphi(r_1,\dots,r_\ell)\in \ML[F\cup G]$ is equivalent to
  an $\ML[G]$-formula $\psi(r_1,\dots,r_\ell)$ of size at most
  $\kappa(K_\varphi)^{\rank_F(\varphi)} \cdot |\varphi|$.
  For notational simplicity, we write $\bar{r}$ for the tuple
  $(r_1,\dots,r_\ell)$. Since $K_\varphi\leq|\varphi|$ and
  $\rank_F(\varphi)\le\log_2(|\varphi|+1)$ by
  Lemma~\ref{thm:bound-rank@lem}, this claim ensures that every
  $\ML[F\cup G]$-formula $\varphi$ of size $n$ has an equivalent
  $\ML[G]$-formula of size at most
  $ \kappa(K_\varphi)^{\log_2(|\varphi|+1)}\cdot |\varphi| \leq
  \kappa(n)^{1+\log_2 n}\cdot n=\kappa'(n)$.
  
  The proof of the claim proceeds by induction on the $F$-rank of $\varphi$.
  Since $F$ remains fixed, we will refer to the $F$-rank simply as rank.
 
  Let $\varphi\in \ML[F\cup G]$ be of rank at most one. We show by
  induction on the structure of $\varphi$ that there exists an
  equivalent $\ML[G]$-formula $\varphi'(\bar{r})$ of size
  $|\varphi'|\leq \kappa(K_\varphi)\cdot |\varphi|$.  If $\varphi$ is
  a propositional variable, 
  $|\varphi|=1\leq\kappa(0)\leq\kappa(1)\cdot|\varphi|$ by choice of the function $\kappa$.
  
  Now, assume that
  $\varphi(\bar{r})=g\formulabrackets{\alpha_1,\ldots,\alpha_k}$
  where $g\in G$ is of arity $k$.
  By induction hypothesis, there exists for each $j\in[k]$ an
  $\ML[G]$-formula $\beta_j(\bar{r})$ with $\beta_j\equiv\alpha_j$ and
  $|\beta_j|\leq\kappa(K_{\alpha_j})\cdot|\alpha_j|$.
  Since $\kappa$ is monotone and $K_{\alpha_j}\leq K_\varphi$, we
  obtain $|\beta_j|\leq\kappa(K_\varphi)\cdot |\alpha_j|$ for
  $j\in[k]$.  Set
  $\varphi'(\bar{r})=g\formulabrackets{\beta_1,\dots,\beta_k}$. Then
  $\varphi'$ is an $\ML[G]$-formula and equivalent to $\varphi$. Furthermore,
  the  size of $\varphi'$ satisfies
  \begin{align*}
    |\varphi'| &=    1+|\beta_1|+\dots+|\beta_k| \\
               &\leq 1+\kappa(K_\varphi)\cdot \bigl(|\alpha_1|+\dots+|\alpha_k|\bigr) \\
               &\leq \kappa(K_\varphi)\cdot \bigl(1+|\alpha_1|+\dots+|\alpha_k|\bigr) 
     && \text{since $\kappa(K_\varphi)\geq 1$}\\
               &=    \kappa(K_\varphi)\cdot |\varphi|\,.
  \end{align*}
  A similar argument establishes the case $\varphi(\overline{r})=\nxt\alpha$.
  
  Finally, assume that
  $\varphi(\bar{r})=f\formulabrackets{\alpha_1,\ldots,\alpha_k}$ where
  $f\in F$ is of arity $k$. Since $F$ and $G$ are disjoint and since all 
  constant functions belong to $G$, we get $k\ge1$. Recall that $\varphi$
  is of rank one and therefore a formula in $N_{F,G}$. Hence there is
  $i\in[k]$ such that, from among the arguments $\alpha_1,\dots,\alpha_k$, 
  at most $\alpha_i$ contains a function from $F$, i.e., such that 
  $\alpha_j\in \ML[G]$ for all $j\in[k]\setminus\{i\}$.
  By induction hypothesis, there is an $\ML[G]$-formula
  $\beta_i(\bar{r})\equiv\alpha_i(\bar{r})$ of size
  $|\beta_i|\leq\kappa(K_{\alpha_i})\cdot|\alpha_i|
  \leq\kappa(K_\varphi)\cdot|\alpha_i|$.
  Recall that $\omega_{f,i}\formulasubst{p_1,\ldots,p_k}$ is a
  $\PL[G]$-representation of $(f,i)$ that uses the variable $p_i$ at most
  once and has size $|\omega_{f,i}|\leq\kappa(k)\leq\kappa(K_\varphi)$.
  Set $\varphi'(\bar{r})=\omega_{f,i}\formulasubst{\beta_1,\ldots,\beta_k}$
  where $\beta_j(\bar{r})=\alpha_j(\bar{r})$ for $j\in[k]\setminus\{i\}$.
  Then $\varphi'$ is an $\ML[G]$-formula and equivalent to
  $f\formulabrackets{\alpha_1,\ldots,\alpha_k}=\varphi$. Furthermore,
  $\varphi'$ is obtained from
  $\omega_{f,i}\formulasubst{p_1,\ldots,p_k}$ by replacing one
  variable with $\beta_i$ and all others with formulae of size at most
  $\sum_{j\neq i}|\alpha_j|$. Hence
  \begin{align*}
    |\varphi'| &= \left|\omega_{f,i}\formulasubst{\beta_1,\ldots,\beta_k}\right| \\ 
            &\leq |\beta_i| + \kappa(k) \cdot \Bigl( 1+\sum_{j\neq i}|\alpha_j| \Bigr) \\
            &\leq \kappa(K_\varphi)\cdot|\alpha_i| + \kappa(K_\varphi)\cdot\Bigl( 1+\sum_{j\neq i}|\alpha_j| \Bigr) \\
               &= \kappa(K_\varphi)\cdot\Bigl(1+\sum_j|\alpha_j| \Bigr) \\ 
               &= \kappa(K_\varphi)\cdot|f\formulabrackets{\alpha_1,\ldots,\alpha_k}| 
                = \kappa(K_\varphi)\cdot|\varphi|.
  \end{align*}
  This shows the claim for formulae of rank at most one. 

  We proceed by induction on the rank of $\varphi$. Assume
  $\varphi\formulasubst{\overline r}\in \ML[F\cup G]$ is of
  rank $R\geq 2$. Let $\gamma\formulasubst{\bar{r},q_1,\ldots,q_n}$ 
  be the derivative of $\varphi\formulasubst{\bar{r}}$ and let
  $\alpha_1,\ldots,\alpha_n\in N_{F,G}$ with
  $\varphi=\gamma\formulasubst{\bar{r},\alpha_1,\ldots,\alpha_n}$.
  Since each of the formulae $\alpha_1(\bar{r}),\dots,\alpha_k(\bar{r})$
  is of rank one, there are 
  $\beta_1(\bar{r}),\ldots,\beta_n(\bar{r})\in\ML[G]$ with
  $\alpha_i\equiv\beta_i$ and
  $|\beta_i|\leq\kappa(K_{\alpha_i})\cdot|\alpha_i|
   \leq\kappa(K_\varphi)\cdot|\alpha_i|$ for $i\in[k]$.  Let
  $\psi\formulasubst{\overline{r}}
  =\gamma\formulasubst{\overline{r},\beta_1,\ldots,\beta_k}$.  Then
  $\psi$ is equivalent to $\varphi$ and of size
  $|\psi|\leq\kappa(K_\varphi)\cdot|\varphi|$. Intuitively, $\psi$ is
  obtained from $\varphi$ by replacing the ``maximal $\ML[F\cup G]$
  sub-formulae'' of rank one by equivalent $\ML[G]$-formulae. 
  Hence the rank of $\psi$ is equal to the rank of $\gamma$, namely 
  $R-1$. It now follows by induction hypothesis that
  $\psi$ is equivalent to a formula $\varphi'\in \ML[G]$ of size
  $|\varphi'|\leq \kappa(K_\psi)^{R-1} \cdot |\psi| \leq
  \kappa(K_\varphi)^R \cdot |\varphi|$. Since
  $\varphi\equiv\psi\equiv\varphi'$, this finishes the verification of
  the claim from the beginning of this proof.
\end{proof}

From Lemma~\ref{thm:@lem}, we can get the main result of this section,
stating that $\ML[F\cup G]$ is not more succinct than $\ML[G]$,
provided $F$ is a finite set of Boolean functions with
$\PL[G]$-representations.

\begin{prop}\label{thm:mon-poly-succinct@prop}
  Let $F$ and $G$ be sets of Boolean functions such that $F$ is 
  finite and has $\PL[G]$-representations and such that $\lverum,
  \lfalsum\in G$. 
  Then $\ML[F\cup G]$ has polynomial translations in $\ML[G]$.
\end{prop}

Noting that $\ML[F]\subseteq\ML[F\cup G]$, the above implies in
particular that $\ML[F]$ has polynomial translations in $\ML[G]$.  In
view of Example~\ref{E:representations}, it ensures specifically that
$\ML[\dM\cup\{\mathrm{maj}\}]$ has polynomial translations in
$\ML[\dM]$.

\begin{proof}
  The set $F'=F\setminus G$ is finite, has $\PL[G]$-representations,
  and, in addition, $F'$ and $G$ are disjoint with $F\cup G=F'\cup
  G$ and $\lverum,\lfalsum\in G$. 

  Choose $\PL[G]$-representations $\omega_{f,i}$ for all $f\in F'$ and
  $i\in[k]$ (where $k>0$ is the arity of $f$). Since $F'$ is finite,
  there is some constant $c>0$ such that all these formulae are of
  size at most $c$. Let $\kappa:\N\to\N$ be the constant function with
  $\kappa(n)=c$ for all $n\in\N$. By the previous lemma, any
  $\varphi\in\ML[F'\cup G]$ is thus equivalent to an $\ML[G]$-formula
  of size at most
  $\kappa(|\varphi|)^{1+\log_2 |\varphi|}\cdot|\varphi|
   = c^{1+\log_2|\varphi|}\cdot|\varphi|
   = c\cdot|\varphi|^{1+\log_2 c}
   \leq c\cdot|\varphi|^d$ for some $d>0$. Since
  $\ML[F'\cup G]=\ML[F\cup G]$, the claim follows.
\end{proof}

\subsection{Representations of $F$ in $\PL[\dM]$ and $\PL[\extdM]$}
\label{SS-step2}

This section is devoted to the second step of our programme (see
page~\pageref{Page-programme}), i.e., we will construct $\PL[\extdM]$-
and $\PL[\dM]$-representations of arbitrary Boolean functions.

\begin{prop}\label{thm:mon-iff-reduce-to-ems@prop} 
  Let $f$ be a Boolean function of arity $k>0$.  Then
  \begin{enumerate}[(1)]
    \item $f$ has $\PL[\extdM]$-representations and
    \item $f$ has $\PL[\dM]$-representations if, and only if, the
      function $f$ is locally monotone.
  \end{enumerate}
\end{prop}

\begin{proof}
  (1) For notational simplicity, we prove that there is some
  $\PL[\extdM]$-representation of $(f,k)$, i.e., that there is some
  $\PL[\extdM]$-formula $\omega(p_1,\ldots,p_k)$ that is equivalent to
  $f\formulabrackets{p_1,\ldots,p_k}$ and that uses the variable $p_k$ 
  at most once. Over the De~Morgan basis there is a formula in disjunctive
  normal form that is equivalent to $f\formulabrackets{p_1,\dots,p_k}$ (but that
  may use $p_k$ more than once). Without loss of generality however,
  we can assume that every clause contains precisely one of $p_k$ and
  $\lnot p_k$. Thus, there are formulae
  $\alpha=\alpha(p_1,\dots,p_{k-1})$ and
  $\beta=\beta(p_1,\dots,p_{k-1})$ in disjunctive normal form such
  that
  \[
    f\formulabrackets{p_1,\dots,p_k} \equiv
    (p_k\land\alpha)\lor(\lnot p_k\land\beta)\,.
  \]
  But this latter formula is equivalent to the formula
  \[
    (\alpha\land\beta)\lor
    \bigl(\lnot(\alpha\land\beta)\land(p_k\liff(\alpha\land\lnot\beta))\bigr)
  \]
  that belongs to $\ML[\extdM]$ and uses $p_k$ only once.
 
  (2) First, let $f$ be locally monotone.  For notational
  simplicity, we will only construct a $\PL[\dM]$-representation of
  $(f,k)$. In addition, we assume that $f$ is increasing in the $k$-th
  argument, i.e., $f(\overline a,\false)\le f(\overline a,\true)$ for
  all $\overline a \in\Bool^{k-1}$.  Since $\dM$ is complete, there is
  a $\PL[\dM]$-formula $\omega\formulasubst{p_1,\dots,p_k}$ that is
  equivalent to $f\formulabrackets{p_1,\dots,p_k}$. Since
  $f(\overline a,\false)\le f(\overline a,\true)$ for any
  $\overline a\in\Bool^{k-1}$, it follows that
  \[ 
    f\formulabrackets{p_1,\ldots,p_k} \equiv 
    \Bigl(\omega\formulasubst[\big]{p_1,\ldots,p_{k-1},\lverum} \land p_k \Bigr)
    \lor \omega\formulasubst[\big]{p_1,\ldots,p_{k-1},\lfalsum}\,.
  \]
  In particular, the formula on the right uses the variable $p_k$ only
  once and therefore forms a $\PL[\dM]$-representation of
  $(f,k)$.

  Now assume that $f$ has $\PL[\dM]$-representations. We prove by
  induction on the size of a formula
  $\omega\formulasubst{p_1,\dots,p_k} \in\PL[\dM]$ the following: if
  the variable $p_k$ appears at most once in $\omega$, then the
  function $\sem{\omega}{p_1,\dots,p_k}$ represented by $\omega$ is
  monotone in its $k$-th argument. The claim is trivial for formulae
  of the form $\lverum\formulabrackets{}$,
  $\lfalsum\formulabrackets{}$, and $p_i$. In particular, we can
  assume $k>0$.

  For the induction step, let $\omega=\alpha_1\lor\alpha_2$. Since
  the formula $\omega$ contains the variable $p_k$ at most once, it
  appears in at most one of the arguments $\alpha_i$; for notational
  simplicity, we assume it does not appear in $\alpha_1$. By the
  induction hypothesis, we have one of the following:
  \begin{itemize}
  \item for all $\overline{a}\in\Bool^{k-1}$:
    $\sem{\alpha_2}{p_1,\dots,p_k}(\overline{a},\false) \leq
    \sem{\alpha_2}{p_1,\dots,p_k}(\overline{a},\true)$ and therefore
    \\ \hspace*{1.9cm}
    $\sem{\alpha_1\lor\alpha_2}{p_1,\dots,p_k}(\overline{a},\false) \le
    \sem{\alpha_1\lor\alpha_2}{p_1,\dots,p_k}(\overline{a},\true)$ or
  \item for all $\overline{a}\in\Bool^{k-1}$:
    $\sem{\alpha_2}{p_1,\dots,p_k}(\overline{a},\false) \geq
    \sem{\alpha_2}{p_1,\dots,p_k}(\overline{a},\true)$ and therefore
    \\ \hspace*{1.9cm}
    $\sem{\alpha_1\lor\alpha_2}{p_1,\dots,p_k}(\overline{a},\false) \ge
    \sem{\alpha_1\lor\alpha_2}{p_1,\dots,p_k}(\overline{a},\true)$.
  \end{itemize}
  In either case, the formula $\omega$ represents a function that is
  locally monotone in its last argument (provided $\omega$ is of the
  form $\alpha_1\lor\alpha_2$).  The remaining cases
  $\omega=\alpha_1\land\alpha_2$ or $\omega=\lnot\alpha_1$ follow by
  a similar argument.
  This finishes the inductive proof.
  
  Recall that $f$ has $\PL[\dM]$-representations. Since each
  $\PL[\dM]$-representation of $(f,i)$ describes a function (namely
  $f$) that is monotone in the $i$-th argument, we obtain that the
  function $f$ is locally monotone, which completes the proof.
\end{proof}

\subsection{Representations of $\dM$ and $\extdM$ in $\PL[G]$}
\label{SS-step3}

This section is devoted to the third step of our programme (see
page~\pageref{Page-programme}), i.e., given an arbitrary complete set
$G$ of Boolean functions, we will construct $\PL[G]$-representations
of the Boolean functions from $\dM$ and $\extdM$,
respectively. More precisely, we show the following.

\begin{thm}\label{T:PL[G]-representations-of-extDM}
  Let $G$ be a complete set of Boolean functions with 
  $\lverum,\lfalsum\in G$.
  \begin{itemize}
  \item Then $\dM$ has $\PL[G]$-representations.
  \item If $G$ contains some non-locally monotone function,
    then $\extdM$ has $\PL[G]$-representations.
  \end{itemize}
\end{thm}
\begin{proof}
  The claims follow from Propositions~\ref{P-generate-negation},
  \ref{P-generate-disjunction}, and \ref{P-generate-biimplication}
  below, as well as  De~Morgan's law
  $x\land y\equiv \lnot(\lnot x\lor\lnot y)$.
\end{proof}

Let $G$ be some complete set of Boolean functions with $\lfalsum,
\lverum\in G$.  We now provide $\PL[G]$-representations for
each of the Boolean functions $\lnot$, $\lor$, and
$\liff$. These constructions follow the proof from \cite{PelM90} (at
least for disjunction, but also the handling of bi-implication is a
variant of their handling of disjunction). Thus, in some sense, the
following constructions can be understood as an analysis of the
constructions by Pelletier and Martin with an additional emphasis on
the condition that the variable $p_1$ is used only once.

\paragraph{Negation}
We start with the construction of a $\PL[G]$-representation of
$(\lnot,1)$.
Since $G$ is complete, there is some $\PL[G]$-formula $\alpha
\formulasubst{p_1}$ that is equivalent to $\lnot p_1$. However, 
$\alpha$ may use $p_1$ more than once, say, $\ell$ times. Let 
$q_1,\dots,q_\ell$ be distinct variables and let $\alpha'\formulasubst{
q_1,\dots,q_\ell}$ be obtained from $\alpha$ by substituting each
occurrence of $p_1$ by one of the variables $q_1,\dots,q_\ell$ such
that each $q_i$ is used exactly once. Then $\alpha$ and $\alpha'$
may not be equivalent but
\[ 
  \sem{\alpha'}{q_1,\dots,q_\ell}(0,\dots,0)=\lnot(0)=1
  \quad\text{and}\quad
  \sem{\alpha'}{q_1,\dots,q_\ell}(1,\dots,1)=\lnot(1)=0\,.
\]
Now, consider the sequence
\begin{align*}
  1=\sem{\alpha'}{q_1,\dots,q_\ell}&(0,0,\dots,0,0),\\
    \sem{\alpha'}{q_1,\dots,q_\ell}&(1,0,\dots,0,0),\\
                                   &\dots           \\
    \sem{\alpha'}{q_1,\dots,q_\ell}&(1,1,\dots,1,0),\\
    \sem{\alpha'}{q_1,\dots,q_\ell}&(1,1,\dots,1,1)=0\,,
\end{align*}
starting with a value of $1$ and ending with a value of $0$.
Hence there is $i\in[\ell]$ such that 
\[
   \sem{\alpha'}{q_1,\dots,q_\ell}(
     \underbrace{1,\dots,1}_{i-1\text{ times}},
     0,
     \underbrace{0,\dots,0}_{\ell-i\text{ times}})
   =1>0=
   \sem{\alpha'}{q_1,\dots,q_\ell}(
     \underbrace{1,\dots,1}_{i-1\text{ times}},
     1,
     \underbrace{0,\dots,0}_{\ell-i\text{ times}})\,.
\]
In particular, $\sem{\alpha'}{q_1,\dots,q_\ell}(1,\dots,1,a,0,
\dots,0)=1-a$, where the $i$-th argument is set to $a\in\Bool$.
Set
\[
  \gamma_j\formulasubst{} =
  \begin{cases}
    \lverum\formulabrackets{}   & \text{if }j<i\\
    \lfalsum\formulabrackets{}  & \text{if }j>i
  \end{cases}
\]
for all $j\in[\ell]\setminus\{i\}$ and consider the $\PL[G]$-formula
\[
    \formulanegation\formulasubst{p_1}
  = \alpha'\formulasubst[\big]{\gamma_1,\dots,\gamma_{i-1},
                               p_1,
                               \gamma_{i+1},\dots,\gamma_\ell}
\]
that uses $p_1$ only once since $\alpha'$ uses no variable twice.
Then, for any interpretation $\Intr$, we get
\begin{align*}
    \Intr(\formulanegation)
  &= \sem{\alpha'}{q_1,\dots,q_\ell}\bigl(
       \Intr(\gamma_1),\dots,\Intr(\gamma_{i-1}),
       \Intr(p_1),
       \Intr(\gamma_{i+1}),\dots,\Intr(\gamma_k)\bigr)\\
  &= \sem{\alpha'}{q_1,\dots,q_\ell}\bigl(1,\dots,1,\Intr(p_1),0,\dots,0\bigr)\\
  &=1-\Intr(p_1)\\
  &=\Intr\bigl(\lnot p_1\bigr)\,.
\end{align*}
Hence the formulae $\formulanegation\formulasubst{p_1}$ and
$\lnot p_1$ are equivalent. 
This shows
\begin{prop}\label{P-generate-negation}
  Let $G$ be a complete set of Boolean functions with $\lfalsum,
  \lverum\in G$. Then $\formulanegation\formulasubst{p_1}$ is a
  $\PL[G]$-representation of $(\lnot,1)$.
\end{prop}

For the next step, we require the following concept. A Boolean
function $f:\Bool^k\to\Bool$ is \emph{affine} if there are 
$c_0,\dots,c_k\in\Bool$ such that $f(a_1,\dots,a_k)=\bigl(c_0+
\sum_{i\in[k]}c_i \cdot a_i\bigr) \bmod 2$ for all $a_1,\dots,
a_k\in\Bool$. Then $f$ is affine if, and only if, for all 
$i\in[n]$, one of the following holds:
\begin{itemize}
\item $f(a_1,\dots,a_{i-1},\false,a_{i+1},\dots,a_k) =
       f(a_1,\dots,a_{i-1},\true, a_{i+1},\dots,a_k)$ for all
      $a_1,\dots,a_n\in\Bool$ or
\item $f(a_1,\dots,a_{i-1},\false,a_{i+1},\dots,a_k) \neq
       f(a_1,\dots,a_{i-1},\true, a_{i+1},\dots,a_k)$ for all
      $a_1,\dots,a_n\in\Bool$.
\end{itemize}
Note that the functions $\lnot$, $\lfalsum$,
$\lverum$, and $\liff$ are affine while $\lor$ and $\land$ are
not affine.

\paragraph{Disjunction}

We now construct a $\PL[G]$-representation of $(\lor,1)$, i.e., a
$\PL[G]$-formula that is equivalent to $p_1\lor p_2$ and uses $p_1$ at
most once. In our construction, we will use the formula
$\formulanegation\formulasubst{p_1}$ from Proposition~\ref{P-generate-negation}.

Since $G$ is complete, there exists a $\PL[G]$-formula $\alpha
\formulasubst{p_1,p_2}$ that is equivalent to $p_1\lor p_2$. Let $\ell$ 
be the number of occurrences of $p_1$ in $\alpha$ and let $m$ be the
number of occurrences of $p_2$ in $\alpha$. Furthermore, let $q_1,\dots,
q_{\ell+m}$ be distinct fresh variables and let $\alpha'\formulasubst{
q_1,\dots,q_{\ell+m}}$ be obtained from $\alpha$ by substituting each 
occurrence of $p_1$ by one of $q_1,\dots,q_\ell$ and each occurrence of 
$p_2$ by one of $q_{\ell+1},\dots,q_{\ell+m}$ such that each new variable 
is used exactly once. For convenience, let $k=\ell+m$ and $f_{\alpha'}$ 
denote the $k$-ary Boolean function $\sem{\alpha'}{q_1,\dots,q_{k}}$.

For $i\in\N$, let $\bar0_i$ and $\bar1_i$ denote the $i$-tuples consisting
solely of zeroes or solely of ones. Then
\begin{align*}
     f_{\alpha'}(\bar0_\ell,\bar0_m)&=\lor(0,0)=0\,,
  &&&f_{\alpha'}(\bar1_\ell,\bar0_m)&=\lor(1,0)=1\,,\\
     f_{\alpha'}(\bar0_\ell,\bar1_m)&=\lor(0,1)=1\,,\text{ and}
  &&&f_{\alpha'}(\bar1_\ell,\bar1_m)&=\lor(1,1)=1\,.
\end{align*}

Aiming at a contradiction, assume that $f_{\alpha'}$ were affine, i.e., 
that there are $c_i\in\Bool$ for $i\in[k]\cup\{0\}$ such that for
all $a_1,\dots,a_{k}\in\Bool$
\[ 
  f_{\alpha'}(a_1,\dots,a_{k}) 
  =
  \Bigl( c_0 + \sum_{i\in[k]} c_i\cdot a_i \Bigr) \bmod 2\,.
\]
Using the relationship between $f_{\alpha'}$ and $\lor$, we can thus
readily observe the following:
\begin{itemize}
  \item $c_0=0$ since $f_{\alpha'}(\bar0_\ell,\bar0_m)=0$,
  \item the set $\{i\in[k] : i\leq\ell \text{ and } c_i=1\}$ 
        is of odd size since $f_{\alpha'}(\bar1_\ell,\bar0_m)=1$ and 
        $c_0=0$, and
  \item the set $\{i\in[k] : i>\ell \text{ and } c_i=1\}$
        is also of odd size since $f_{\alpha'}(\bar0_\ell,\bar1_m)=1$
        and $c_0=0$.
\end{itemize}
But then the total number of non-zero coefficients is even, hence 
$\lor(1,1)=f_{\alpha'}(\bar1_\ell,\bar1_m)=0$, a contradiction. 
Consequently, $f_{\alpha'}$ is not affine and there
are $i\in[k]$ and $a_1,\dots,a_k,b_1,\dots,b_k\in\Bool$ such that
\begin{align*}
    f_{\alpha'}(a_1,\dots,a_{i-1},0,a_{i+1},\dots,a_k)
  &=f_{\alpha'}(a_1,\dots,a_{i-1},1,a_{i+1},\dots,a_k)\text{ and }\\
    f_{\alpha'}(b_1,\dots,b_{i-1},0,b_{i+1},\dots,b_k)
  &\neq f_{\alpha'}(b_1,\dots,b_{i-1},1,b_{i+1},\dots,b_k)\,.
\end{align*}
Recall that we want to construct a representation of $\lor$ that uses
$p_1$ only once. The general idea is as follows: the variable $q_i$
will be substituted by $p_1$ (or $\formulanegation(p_1)$, depending on 
the precise distribution of zeroes and ones in the above equalities).
At this point it is important that $\alpha'$ uses $q_i$ only once.
The remaining variables $q_j$ with $j\neq i$ will be substituted by 
formulae depending solely on $q_2$ such that, if $q_2$ is set to one, 
we arrive in the first line and if $q_2$ is set to zero, we arrive in
the second line.
More precisely, we define formulae $\gamma_j$ for $j\in[k]$ as follows:
\[
  \gamma_j\formulasubst{p_1,p_2}=
  \begin{cases}
    \lverum\formulabrackets{}          & \text{if }a_j=b_j=1 \text{ and } j\neq i\\
    \lfalsum\formulabrackets{}         & \text{if }a_j=b_j=0 \text{ and } j\neq i\\
    \formulanegation\formulasubst{p_2} & \text{if }a_j<b_j   \text{ and } j\neq i\\
    p_2                                & \text{if }a_j>b_j   \text{ and } j\neq i\\
    p_1                                & \text{if }j=i
  \end{cases}
\]
Then each $\gamma_j\formulasubst{p_1,p_2}\in\PL[G]$ uses both $p_1$
and $p_2$ at most once. Furthermore, there is at most one $j\in[k]$
such that $\gamma_j$ uses $p_1$ at all ($p_2$ can be used by more 
than one of these formulae since there may exist several $j\neq i$
with $a_j\neq b_j$).
Then these formulae have indeed the properties described above, as
for any interpretation $\Intr$ we have
\[
  \Intr(\gamma_j)=
  \begin{cases}
    a_j & \text{if }\Intr(p_2)=1\text{ and }j\neq i\,,\\
    b_j & \text{if }\Intr(p_2)=0\text{ and }j\neq i\,,\\
    \Intr(p_1) & \text{if }j=i\,.
  \end{cases}
\]
Let $\formulapredisjunction\formulasubst{p_1,p_2}=
\alpha'\formulasubst{\gamma_1,\dots,\gamma_k}$ and, for $a,b\in\Bool$,
let $\Intr_{a,b}$ denote an interpretation with $\Intr_{a,b}(p_1)=a$
and $\Intr_{a,b}(p_2)=b$. Then 
\begin{align*}
        \Intr_{0,1}(\formulapredisjunction) 
  &=    f_{\alpha'}(a_1,\dots,a_{i-1},0,a_{i+1},\dots,a_k) \\
  &=    f_{\alpha'}(a_1,\dots,a_{i-1},1,a_{i+1},\dots,a_k)
   =    \Intr_{1,1}(\formulapredisjunction)
\intertext{and}
        \Intr_{0,0}(\formulapredisjunction)
  &=    f_{\alpha'}(b_1,\dots,b_{i-1},0,b_{i+1},\dots,b_k) \\
  &\neq f_{\alpha'}(b_1,\dots,b_{i-1},1,b_{i+1},\dots,b_k)
   =    \Intr_{1,0}(\formulapredisjunction)
\end{align*}
Furthermore, $\formulapredisjunction=\alpha'\formulasubst{\gamma_1,\dots,\gamma_k}\in\PL[G]$
uses $p_1$ only once since $\alpha'$ uses no variable twice and
since at most one of the formulae $\gamma_j$ uses $p_1$ at all.

Now, all that remains is to modify $\formulapredisjunction$ by negating,
if necessary, $p_1$, the whole formula, or both, such that
$\Intr_{1,1}(\formulapredisjunction')=1$ and 
$\Intr_{0,0}(\formulapredisjunction')=0$. The above equalities then imply
$\Intr_{0,1}(\formulapredisjunction')=1$ and 
$\Intr_{1,0}(\formulapredisjunction')=1$.
We accomplish this in two steps:
\begin{enumerate}
\item If $\Intr_{1,1}(\formulapredisjunction)=1$, then set 
  $\formulapredisjunction_1=\formulapredisjunction$. Otherwise, 
  set $\formulapredisjunction_1=\formulanegation\formulasubst{
  \formulapredisjunction}$, i.e., we obtain $\formulapredisjunction_1$ 
  from the formula $\formulanegation\formulasubst{p_1}$ by substituting
  $p_1$ by $\formulapredisjunction$.%
  \footnote{Since $\Intr(\formulapredisjunction)$ depends solely on
       $\Intr(p_1)$ and $\Intr(p_2)$, the formula $\formuladisjunction$
       is well-defined.} 
  In this latter case, we obtain
  \begin{align*}
          \Intr_{0,1}(\formulapredisjunction_1)
    &=    1-\Intr_{0,1}(\formulapredisjunction)
     =    1-\Intr_{1,1}(\formulapredisjunction)
     =    \Intr_{1,1}(\formulapredisjunction_1)
    \text{ and}\\
          \Intr_{0,0}(\formulapredisjunction_1)
    &=    1-\Intr_{0,0}(\formulapredisjunction)
     \neq 1-\Intr_{1,0}(\formulapredisjunction)
     =    \Intr_{1,0}(\formulapredisjunction_1)\,.
  \end{align*}
  Since $\Intr_{1,1}(\formulapredisjunction)=1$, this implies
  \begin{align*}
          \Intr_{0,1}(\formulapredisjunction_1)
    &=    \Intr_{1,1}(\formulapredisjunction_1)=1 \text{ and}\\
          \Intr_{0,0}(\formulapredisjunction_1)
    &\neq \Intr_{1,0}(\formulapredisjunction_1)\,.
  \end{align*}
  Note that this also holds in case $\Intr_{1,1}(\formulapredisjunction)=1$.

\item If $\Intr_{0,0}(\formulapredisjunction_1)=0$, then set
  $\formulapredisjunction_2=\formulapredisjunction_1$. Otherwise, set
  $\formulapredisjunction_2(p_1,p_2) =
  \formulapredisjunction_1(\formulanegation\formulasubst{p_1},p_2))$,
  i.e., we obtain $\formulapredisjunction_2$ from the formula
  $\formulapredisjunction_1$ by replacing all occurrences of $p_1$ by
  the formula $\formulanegation\formulasubst{p_1}$. Again, since
  $\Intr(\formulapredisjunction_1)$ depends, at most, on $\Intr(p_1)$
  and $\Intr(p_2)$, the formula $\formulapredisjunction_2$ is
  well-defined. In the latter case, we obtain
  \begin{align*}
         \Intr_{1,1}(\formulapredisjunction_2)
    &=   \Intr_{0,1}(\formulapredisjunction_1)
     =   \Intr_{1,1}(\formulapredisjunction_1)
     =   \Intr_{0,1}(\formulapredisjunction_2)
    \text{ and}\\
         \Intr_{0,0}(\formulapredisjunction_2)
    &=   \Intr_{1,0}(\formulapredisjunction_1)
     \neq\Intr_{0,0}(\formulapredisjunction_1)
     =   \Intr_{1,0}(\formulapredisjunction_2)\,.
  \end{align*}
  Since $\Intr_{0,0}(\formulapredisjunction_1)=1$, this implies
  \begin{align*}
         \Intr_{0,1}(\formulapredisjunction_2)
    &=   \Intr_{1,1}(\formulapredisjunction_2)=1
    \text{ and}\\
    0=   \Intr_{0,0}(\formulapredisjunction_2)
    &\neq\Intr_{1,0}(\formulapredisjunction_2)=1\,.
  \end{align*}
  Note that this also holds in case
  $\Intr_{0,0}(\formulapredisjunction_1)=0$.
\end{enumerate}

In summary, we have $\Intr_{a,b}(\formulapredisjunction_2)=1$ if, and
only if, $a=1$ or $b=1$. Hence, indeed, the formulae $p_1\lor p_2$ and
$\formulapredisjunction_2\formulasubst{p_1,p_2}$ are
equivalent. Since, in the above procedure, we did not duplicate any
variables, the formula $\formulapredisjunction_2$ uses the variable
$p_1$ at most once. Hence choosing $\formuladisjunction\formulasubst
{p_1,p_2}=\formulapredisjunction_2\formulasubst{p_1,p_2}$ proves
\begin{prop}\label{P-generate-disjunction}
  Let $G\ni\lverum,\lfalsum$ be a complete set of Boolean
  functions. Then $\formuladisjunction$ is a $\PL[G]$-representation
  of $(\lor,1)$; by symmetry, there is also a $\PL[G]$-representation
  of $(\lor,2)$.
\end{prop}

\paragraph{Bi-implication}

We next construct a $\PL[G]$-representation of $(\liff,1)$, i.e., a
$\PL[G]$-formula that is equivalent to $p_1\liff p_2$ and uses $p_1$
at most once. Here, we use the additional assumption that $G$ contains
some function $g$ that is not locally monotone. In our construction,
we will use the $\PL[G]$-representations 
$\formulanegation\formulasubst{p_1}$ and
$\formuladisjunction\formulasubst{p_1,p_2}$ of $(\lnot,1)$ and
$(\lor,1)$ from Propositions~\ref{P-generate-negation} and
\ref{P-generate-disjunction}, respectively.

Since $g$ is not locally monotone, its arity $k$ is at least two and
there are $i\in[k]$ and $a_1,\dots,a_k,b_1,\dots,b_k\in\Bool$ such that
\begin{align*}
  g(a_1,\dots a_{i-1},0,a_{i+1},\dots,a_k)
  &<g(a_1,\dots a_{i-1},1,a_{i+1},\dots,a_k)\text{ and }\\
  g(b_1,\dots b_{i-1},0,b_{i+1},\dots,b_k)
  &>g(b_1,\dots b_{i-1},1,b_{i+1},\dots,b_k)\,.
\end{align*}
Let the formulae $\gamma_j$ be defined as in the previous section.
We consider the formula
\[
  \formulabiimplication\formulasubst{p_1,p_2}=
  g\formulabrackets[\big]{\gamma_1,\gamma_2,\dots,\gamma_k}\,.
\]
Since at most one of the formulae $\gamma_j$ uses $p_1$ at all, and
this formula uses $p_1$ at most once, the formula $\formulabiimplication$
uses $p_1$ at most once.

For $a,b\in\Bool$, let $\Intr_{a,b}$ denote an interpretation
with $\Intr_{a,b}(p_1)=a$ and $\Intr_{a,b}(p_2)=b$. 
Then we have
\begin{align*}
  \Intr_{0,1}(\formulabiimplication)
  &= g\bigl(\Intr_{0,1}(\gamma_1),\dots,\Intr_{0,1}(\gamma_k)\bigr)\\
  &= g(a_1,\dots,a_{i-1},0,a_{i+1},\dots,a_k)\\
  &< g(a_1,\dots,a_{i-1},1,a_{i+1},\dots,a_k)\\
  &= g\bigl(\Intr_{1,1}(\gamma_1),\dots,\Intr_{1,1}(\gamma_k)\bigr)\\
  &= \Intr_{1,1}(\formulabiimplication)\\
    \intertext{and }
  \Intr_{0,0}(\formulabiimplication)
  &= g\bigl(\Intr_{0,0}(\gamma_1),\dots,\Intr_{0,0}(\gamma_k)\bigr)\\
  &= g(b_1,\dots,b_{i-1},0,b_{i+1},\dots,b_k)\\
  &> g(b_1,\dots,b_{i-1},1,b_{i+1},\dots,b_k)\\
  &= g\bigl(\Intr_{1,0}(\gamma_1),\dots,\Intr_{1,0}(\gamma_k)\bigr)\\
  &= \Intr_{1,0}(\formulabiimplication)
\end{align*}
It follows that
\[
  \Intr_{a,b}(\formulabiimplication)=1\iff a=b\,,
\]
which proves
\begin{prop}\label{P-generate-biimplication}
  Let $G\ni\lverum,\lfalsum$ be a complete set of Boolean functions
  with a non-locally monotone function $g\in G$. Then
  $\formulabiimplication\formulasubst{p_1,p_2}$ is a
  $\PL[G]$-representation of $(\liff,1)$; by symmetry, there is also a
  $\PL[G]$-representation of $(\liff,2)$.
\end{prop}

\subsection{Proof of this section's main theorem}
\label{SS-proof-main-theorem}

\begin{proof}[Proof of Theorem~\ref{T-main} (page~\pageref{T-main})]

  First, consider the case that all functions from $F$ are locally
  monotone. By Proposition~\ref{thm:mon-iff-reduce-to-ems@prop}, all
  functions from $F$ have $\PL[\dM]$-representations. Hence, by
  Proposition~\ref{thm:mon-poly-succinct@prop}, $\ML[F]$ has
  polynomial translations in $\ML[\dM]$.
  In the previous section, we derived representations of $\lnot$ and 
  $\lor$ (for complete sets of Boolean functions containing $\lverum$ 
  and $\lfalsum$), but not explicitly of $\land$. However, replacing 
  any sub-formula of the form $\alpha\land\beta$ by
  $\lnot(\lnot\alpha\lor\lnot\beta)$ defines a polynomial translation
  of $\ML[\dM]$ in $\ML[\{\lverum,\lfalsum,\lnot,\lor\}]$. 
  Since the concatenation of polynomial translations is again a polynomial 
  translation, it follows that $\ML[F]$ has polynomial translations in 
  $\ML[\{\lverum,\lfalsum,\lnot,\lor\}]$.

  Now consider $G^+=G\cup\{\lverum,\lfalsum\}$. This set of Boolean
  functions is complete since $G$ is complete. By
  Propositions~\ref{P-generate-negation} and
  \ref{P-generate-disjunction}, the functions $\lnot$ and $\lor$ have
  $\PL[G^+]$-representations. Again by
  Proposition~\ref{thm:mon-poly-succinct@prop},
  $\ML[\{\lverum,\lfalsum,\lnot,\lor\}]$ has polynomial translations
  in $\ML[G^+]$.
  
  Since $G$ is complete, there are formulae
  $\formulatrue,\formulafalse\in\PL[G]$ that are equivalent to
  $\lverum$ and $\lfalsum$, respectively (note that these formulae may
  contain some variables). Replacing any sub-formula $\lverum$ and
  $\lfalsum$ by $\formulatrue$ and $\formulafalse$ defines a
  polynomial translation of $\ML[G^+]$ in $\ML[G]$.
  Using once more that the concatenation of polynomial translations
  is a polynomial
  translation, we find a polynomial translation of $\ML[F]$ in
  $\ML[G]$ provided all functions from $F$ are locally monotone and
  $G$ is complete.

  Now suppose that $G$ contains some function that is not locally
  monotone. Then we can argue as above, using
  \begin{itemize}
  \item the logic $\ML[\extdM]$ instead of the logic $\ML[\dM]$,
  \item the logic $\ML[\{\lverum,\lfalsum,\lnot,\lor,\liff\}]$ instead
    of the logic  $\ML[\{\lverum,\lfalsum,\lnot,\lor\}]$, and
  \item Proposition~\ref{P-generate-biimplication} in addition to the
    Propositions~\ref{P-generate-negation} and
    \ref{P-generate-disjunction}.\qedhere
  \end{itemize}
\end{proof}

Recall that the main result of this section was a simple corollary
of Theorem~\ref{T-main}. Hence we have shown that the class of logics
$\ML[G]$ with $G$ finite and complete has at most two succinctness
classes, namely that of $\dM$ and that of $\extdM\supseteq\dM$.

Now let $\cS$ be any class of pointed Kripke structures.  Since
formulae equivalent wrt.\ \MLK are trivially equivalent wrt.\ $\cS$,
Corollary~\ref{C-two-succinctness-classes} also holds for the
succinctness wrt.~$\cS$ -- which is the restricted meaning of ``all''
in the title of this section. But as the reader will have realized,
the above proof technique carries over to many other logics like
multi-modal logic, predicate calculus, temporal logics etc.\ -- which
gives a wider (and imprecise) meaning to the term ``all''.

\section{Modal logics with one and with two succinctness classes}

So far, we saw that for many logics $\cL$, there are at most two
succinctness classes (up to a polynomial), namely those of $\cL[\dM]$
and $\cL[\extdM]$, respectively. In this section, we will show that
these classes may differ in some modal logics, but coincide in others.
Recall the following classes of pointed Kripke structures with the
corresponding properties of their accessibility relations
\begin{itemize}
  \item $\MLK$, the class of all pointed Kripke structures,
  \item $\MLT$, the class of reflexive pointed Kripke structures,
        and
  \item $\MLSf$, the class of equivalence relations.
\end{itemize}
More precisely, we show that the succinctness classes differ for the
modal logics $\mathrm{T}$ (and hence also for $\mathrm{K}$) but
coincide for the modal logic $\mathrm{S5}$. We will therefore, from
now on, be precise and return to the original notation, i.e., write
$\equiv_{\MLK}$ instead of $\equiv$ etc.

\subsection{Two distinct succinctness classes}

Throughout this section, we will consider the class $\MLT$ of pointed
Kripke structures where the accessibility relation $R$ is reflexive.
We will demonstrate that $\ML[\extdM]$ does not have polynomial
translations in $\ML[\dM]$ with respect to the class $\MLT$.  To this
aim, we consider the following sequence of formulae from
$\ML[\extdM]$:
\[
  \varphi_0  = p_0 \land\nxt\lnot p_0
  \quad\text{ and }\quad
  \varphi_{n+1}  =p_{(n+1)\bmod 2}\land(p\liff\nxt\varphi_n)\,.
\]
Note that there is $c>0$ such that $|\varphi_n|\le c\cdot n$ for all
$n\ge1$. For all $n\in\N$, we will prove that all formulae
$\psi\in\ML[\dM]$ with $\varphi_n \equiv_{\MLT} \psi$ satisfy
$|\psi|\ge 2^n$. Consequently, $\ML[\extdM]$ does not have
sub-exponential translations in $\ML[\dM]$ -- the two succinctness
classes of $\ML[\dM]$ and $\ML[\extdM]$ are distinct.

Actually, we show a bit more: The formula $\psi$ is not only
``large'', but it even contains exponentially many occurrences of the
modal operator $\nxt$. We denote this number of occurrences of $\nxt$
in the formula $\psi$ by $\lnxtcount{\psi}$.

In the course of this proof, we need the following concept. Let
$\psi\in\ML[\dM]$. Then $\psi$ is a Boolean combination of atomic
formulae $p_i$, $\lverum$, and $\lfalsum$ and of formulae of the form
$\nxt\lambda$ with $\lambda\in\ML[\dM]$. The set $E_\psi$ consists of
all formulae $\lambda$ such that $\nxt\lambda$ appears in the Boolean
combination $\psi$ with even negation depth; the set $O_\psi$ consists
of all $\lambda\in\ML[\dM]$ such that $\nxt\lambda$ appears with odd
negation depth. Inductively, these two sets are defined as follows:
\begin{align*}
  E_\psi&=
  \begin{cases}
    \emptyset   & \text{if }\psi\in\{\lverum,\lfalsum\}\cup\Var\\
    \{\lambda\} & \text{if }\psi=\nxt\lambda, \lambda\in\ML[\dM]\\
    O_\alpha    & \text{if }\psi=\lnot\alpha\\
    E_\alpha\cup E_\beta& \text{if }\psi\in\{\alpha\land\beta,\alpha\lor\beta\}
  \end{cases}\\
  O_\psi&=
  \begin{cases}
    \emptyset  & \text{if }\psi\in\{\lverum,\lfalsum\}\cup\Var\\
    \emptyset  & \text{if }\psi=\nxt\lambda, \lambda\in\ML[\dM]\\
    E_\alpha   & \text{if }\psi=\lnot\alpha\\
    O_\alpha\cup O_\beta& \text{if }\psi\in\{\alpha\land\beta,\alpha\lor\beta\}
  \end{cases}
\end{align*}
The use of these sets is described by the following lemma that can be
proved by structural induction.

\begin{lem}\label{L-E-and-O}
  Let $S,T\in\MLK$ and $\psi\in\ML[\dM]$ such that the following
  properties hold:
  \begin{enumerate}[(i)]
  \item $S\models p\iff T\models p$ for all atomic formulae
    $p$. \label{atomic-clause-in-L-E-and-O}
  \item If $\lambda\in E_\psi$ with $S\models \nxt\lambda$, then
    $T\models\nxt\lambda$. \label{E-clause-in-L-E-and-O}
  \item If $\lambda\in O_\psi$ with $T\models \nxt\lambda$, then
    $S\models\nxt\lambda$. \label{O-clause-in-L-E-and-O}
  \end{enumerate}
  Then $S\models\psi$ implies $T\models\psi$.
\end{lem}

Fix, for every formula $\alpha\in\ML[\dM]$, a pointed Kripke structure
$S_\alpha=(W_\alpha,R_\alpha,V_\alpha,\iota_\alpha)$ from $\MLT$ such that
$S_\alpha\models\alpha$ whenever $\alpha$ is satisfiable in $\MLT$
(for unsatisfiable formulae $\alpha$, fix an arbitrary structure
$S_\alpha$ from $\MLT$). We will always assume that these structures
$S_\alpha$ are mutually disjoint. 

We now come to the central result of this section.

\begin{lem}\label{L-psi-large}
  For all $n\in\N$ and all formulae $\psi\in\ML[\dM]$ with
  $\psi \equiv_{\MLT} \varphi_n$, we have $\lnxtcount{\psi}\ge 2^n$.
\end{lem}

\begin{proof}
  Since we want to prove this lemma by induction on $n$, we first
  consider the case $n=0$, i.e., the formula
  $\varphi_0=p_0\land\nxt\lnot p_0$. Let $\psi\in\ML[\dM]$ with
  $\psi \equiv_{\MLT} \varphi_0$. Towards a contradiction, assume
  $\lnxtcount{\psi}<2^n=1$, i.e., assume $\psi$ to be a Boolean
  combination of atomic formulae. Consider the Kripke structure
  $S_0\in\MLT$ from Fig.~\ref{fig:iff-exp-succinct:n=0}. Then
  $(S_0,\kappa)\models\varphi_0$ implies $(S_0,\kappa)\models\psi$ since
  $\varphi_0$ and $\psi$ are assumed to be equivalent. Since $\psi$ is
  a Boolean combination of atomic formulae, and since the worlds
  $\kappa$ and $\iota$ satisfy the same atomic formulae, we obtain
  $(S_0,\iota)\models\psi$ and therefore $(S_0,\iota)\models\varphi_0$,
  which is not the case. Hence we showed $\lnxtcount{\psi}\ge1=2^0$, 
  i.e., we verified the base case $n=0$.

  \begin{figure}
    \centering
    \begin{tikzpicture}[structure]
      \node[world,
            label={[label distance=-2pt]135:{$\substack{\kappa\\ p_0}$}}]  
           (k) {};
      \node[world,below=of k,
            label={[label distance=-2pt]225:{$\substack{v\\\ }$}}]
           (v) {};
      \node[world,right=1.5 of k,
            label={[label distance=-2pt]135:{$\substack{\iota\\ p_0}$}}]
           (i) {};
      \draw (k) edge (v)
            (k) edge[out= 20,in= 60,looseness=14] (k)
            (i) edge[out= 20,in= 60,looseness=14] (i)
            (v) edge[out=290,in=340,looseness=14] (v);
    \end{tikzpicture}
    \caption{Schematic representation of the Kripke structure $S_0$.
             Vertices are labelled with the name of the world and the 
             propositional variables holding there.}
    \label{fig:iff-exp-succinct:n=0}
  \end{figure}

  So let $n\ge0$ and consider the formula
  $\varphi_{n+1}=p_{(n+1)\bmod 2}\land (p\liff\nxt\varphi_n)$. Let
  $\psi\in\ML[\dM]$ be some formula such that
  $\psi \equiv_{\MLT} \varphi_{n+1}$.

  Aiming at a contradiction, assume $\lnxtcount{\psi}<2^{n+1}$.
  Although the sets $E_\psi$ and $O_\psi$ may have non-empty
  intersection, it follows that
  \[
      \sum_{\lambda\in E_\psi} \lnxtcount{\lambda}
    \ +\ 
    \sum_{\lambda\in O_\psi} \lnxtcount{\lambda}
    \le \lnxtcount{\psi} 
    <    2^{n+1}\,.
  \]
  Hence, for at least one of the sets $E_\psi$ and $O_\psi$, the total
  number of modal operators must be less than $2^n$; we consider these two
  cases separately.

  \paragraph{Case 1,
    $\sum_{\lambda\in O_\psi} \lnxtcount{\lambda} < 2^n$}

  Let $O^\ast_\psi$ be the set of formulae $\lambda$ from $O_\psi$
  with $\lambda\models_{\MLT}\varphi_n$. Since
  $\lnxtcount{\bigvee O^\ast_\psi} \le \sum_{\lambda\in O_\psi}
  \lnxtcount{\lambda} < 2^n$, the induction hypothesis ensures
  $\bigvee O^\ast_\psi \not\equiv_{\MLT} \varphi_n$.  On the other
  hand, by choice of $O^\ast_\psi$, we have
  $\bigvee O^\ast_\psi \models_{\MLT} \varphi_n$ (this holds also in
  case $O^\ast_\psi=\emptyset$ since
  $\bigvee\emptyset\equiv_{\MLK}\lfalsum$). Hence
  $\varphi_n\not\models_{\MLT}\bigvee O^\ast_\psi$ implying that the
  formula $\alpha=\varphi_n\land\lnot\bigvee O^\ast_\psi$ is
  satisfiable in $\MLT$. The choice of the structure $S_\alpha$
  implies $S_\alpha\models\alpha$, i.e., $S_\alpha\models\varphi_n$
  but $S_\alpha\not\models\bigvee O^\ast_\psi$.
  Finally, let $B$ denote the set of formulae $\beta$ with
  $S_\beta\models\beta\land\lnot\varphi_n$.
  
  We now define a Kripke structure $S=(W,R,V)$ as follows
  (cf.~Fig.~\ref{fig:trees}):
  \begin{align*}
    W &= \{\iota,\kappa\}\uplus\bigcup_{\beta\in B}W_\beta \cup W_\alpha\\
    R &= \bigl\{(\iota,\iota),(\kappa,\kappa),
                (\iota,\kappa),(\iota,\iota_\alpha) \}
        \cup \Bigl(\{\iota,\kappa\}\times\{ \iota_\beta \mid \beta\in B \}\Bigr)
        \cup \bigcup_{\beta\in B}R_\beta
        \cup R_\alpha\\
    V(q) &=\begin{cases}
       \bigcup_{\beta\in B} V_\beta(q) \cup V_\alpha(q) \cup \{\iota,\kappa\} 
      & \text{if } q=p_{(n+1)\bmod2}\\
      \bigcup_{\beta\in B} V_\beta(q) \cup V_\alpha(q) 
      & \text{otherwise }
             \end{cases}
  \end{align*}
  Since the accessibility relations of the structures $S_\alpha$ and
  $S_\beta$ for $\beta\in H$ are reflexive, the same applies to the
  accessibility relation of $S$, i.e., we obtain
  $S\in\MLT$. Furthermore, 
  \begin{equation}
    \label{eq:substructure}
    \text{for all }\lambda\in\ML[\dM], \gamma\in B\cup\{\alpha\}, w\in W_\gamma\colon \Bigl((S_\gamma,w)\models\lambda \iff (S,w)\models\lambda\Bigr)
  \end{equation}
  since we only add edges originating in $\iota$ or $\kappa$.

  \begin{figure}
    \centering
    \begin{tikzpicture}[structure]
      \node[world,
            label={[label distance=-2pt]135:
                   {$\substack{\iota\\ p_{(n+1)\bmod2}}$}}]  (i) {};
      \node[world,below right=.6cm and 3.1cm of i,
            label={[label distance=-2pt]0:
                   {$\substack{\kappa\\ p_{(n+1)\bmod2}}$}}] (i'){};

      \node[tree1,below left=2.3cm and -0cm of i]    (b1t) {$S_{\beta_1}$};
      \node[tree2,below left=2.3cm and -1.5cm of i]  (b2t) {$S_{\beta_2}$};
      \node[      below left=2.0cm and -2.8cm of i]        {$\dots$};
      \node[tree3,below left=2.3cm and -3.7cm of i]  (brt) {$S_{\beta_r}$};
      \node[      below left=2.0cm and -5cm of i]          {$\dots$};
      \node[tree1,below left=1.3cm and 1.6cm of i]   (at)  {$S_\alpha$};
      
      \node[pin,above=-5pt of b1t] (b1) {};
      \node[pin,above=-5pt of b2t] (b2) {};
      \node[pin,above=-5pt of brt] (br) {};
      \node[pin,above=-5pt of at]  (a) {};
      
      \draw (i)  edge[out= 20,in= 70,looseness=14] (i)
            (i') edge[out= 65,in=115,looseness=14] (i')
            (i) edge (i')
            (i)  edge (b1)
            (i)  edge (b2)
            (i)  edge (br)
            (i)  edge (a) 
            (i') edge[bend right=5] (b1)
            (i') edge (b2)
            (i') edge (br);
    \end{tikzpicture}
    \caption{Schematic representation of the Kripke structure $S$.}
    \label{fig:trees}
  \end{figure}

  We will proceed by verifying the following claims:
  \begin{enumerate}[(A)]
  \item\label{claim-1} $(S,\iota) \not\models \varphi_{n+1}$,
  \item\label{claim-2} $(S,\kappa) \models \varphi_{n+1}$,
  \item\label{claim-3} $(S,\kappa) \models \psi$, and
  \item\label{claim-4} $(S,\iota) \models \psi$,
  \end{enumerate}
  thus showing $(S,\iota)\models\lnot\varphi_{n+1}\land\psi$, which
  contradicts the equivalence of $\varphi_{n+1}$ and $\psi$.

  \paragraph{Proof of~\ref{claim-1},
    $(S,\iota)\not\models\varphi_{n+1}$}  Recall that
  $(S_\alpha,\iota_\alpha)\models\varphi_n$ implying
  $(S,\iota_\alpha)\models\varphi_n$ by \eqref{eq:substructure}. Hence we have
  $(S,\iota)\models \lnot p\land\nxt\varphi_n$ which
  ensures $(S,\iota)\not\models\varphi_{n+1}$.

  \paragraph{Proof of~\ref{claim-2},
    $(S,\kappa)\models\varphi_{n+1}$}

  Towards a contradiction, suppose
  $(S,\kappa)\models\nxt\varphi_n$. Then there exists $w\in W$ with
  $(S,w)\models\varphi_n$ and $(\kappa,w)\in R$, i.e.,
  $w\in\{\kappa\}\cup\{\iota_\beta\mid \beta\in B\}$.  From
  $\varphi_n\models_{\MLK} p_{n\bmod2}$ and
  $(S,\kappa)\not\models p_{n\bmod2}$, we get
  $(S,\kappa)\models\lnot\varphi_n$ and therefore $w\neq\kappa$. Hence
  there is $\beta\in B$ with $w=\iota_\beta$. Now
  $(S,w)\models\varphi_n$ implies $(S_\beta,w)\models\varphi_n$ by
  \eqref{eq:substructure}. But this contradicts $\beta\in B$. Thus,
  indeed, $(S,\kappa)\not\models\nxt\varphi_n$.  Together with the
  observation $(S,\kappa)\models p_{(n+1)\bmod2}\land\lnot p$, we get
  $(S,\kappa)\models\varphi_{n+1}$.

  \paragraph{Proof of~\ref{claim-3}, $(S,\kappa)\models\psi$}
  This is immediate by the above since $S\in\MLT$ and
  $\psi \equiv_{\MLT} \varphi_{n+1}$.

  \paragraph{Proof of~\ref{claim-4}, $(S,\iota)\models\psi$}  From
  $(S,\kappa)\models\psi$, we will now infer that $\psi$ holds in the
  world~$\iota$, too. Recall that $\psi$ is a Boolean combination of
  atomic formulae and of formulae $\nxt\lambda$ with
  $\lambda\in O_\psi\cup E_\psi$. Note that $(S,\iota)$ and
  $(S,\kappa)$ agree in the atomic formulae holding there. In view of
  Lemma~\ref{L-E-and-O}, it therefore suffices to show the following:
  \begin{enumerate}[(D1)]
  \item\label{claim-4-1} If $\lambda\in E_\psi$ with
    $(S,\kappa)\models\nxt\lambda$, then
    $(S,\iota)\models\nxt\lambda$.
  \item\label{claim-4-2} If $\lambda\in O_\psi$ with
    $(S,\iota)\models\nxt\lambda$, then
    $(S,\kappa)\models\nxt\lambda$.
  \end{enumerate}

  \paragraph{Proof of~\ref{claim-4-1}}
  Let $\lambda\in E_\psi$ with $(S,\kappa)\models\nxt\lambda$.  Then
  there is $w\in\{\kappa\}\cup\{\iota_\beta\mid\beta\in B\}$ such that
  $(S,w)\models\lambda$. In any case, $(\iota,w)\in R$. Hence we
  have $(S,\iota)\models\nxt\lambda$.

  \paragraph{Proof of~\ref{claim-4-2}}
  Let $\lambda\in O_\psi$ with $(S,\iota)\models\nxt\lambda$.

  First, assume $\lambda\models_{\MLT}\varphi_n$, i.e.,
  $\lambda\in O^\ast_\psi$.  From $(S,\iota)\models\nxt\lambda$, we
  obtain that the formula $\lambda$ holds in one of the worlds
  $\iota$, $\kappa$, $\iota_\alpha$, or $\iota_\beta$ for some
  $\beta\in B$. But $(S,\iota_\alpha)\models\lambda$ implies
  $S_\alpha\models\lambda$ by \eqref{eq:substructure}, which is
  impossible since $S_\alpha\models\alpha$,
  $\alpha\models_{\MLK} \lnot\bigvee O^\ast_\psi$, and
  $\lnot\bigvee O^\ast_\psi\models_{\MLK} \lnot\lambda$ since
  $\lambda\in O^\ast_\psi$. Next suppose
  $(S,\iota_\beta)\models\lambda$ for some $\beta\in B$. Then
  $S_\beta\models\lambda$ which, together with
  $\lambda\models_{\MLT}\varphi_n$, implies
  $S_\beta\models\varphi_n$. Because of $\beta\in B$, we have
  $S_\beta\models\lnot\varphi_n$, a contradiction.

  Consequently, the formula $\lambda$ holds in one of the worlds
  $\iota$ and $\kappa$. Again using $\lambda\models_{\MLT}\varphi_n$,
  we obtain that also $\varphi_n$ holds in $\iota$ or in $\kappa$. But
  this cannot be the case since $p_{n\bmod 2}$ does not hold in either
  of the two worlds -- a contradiction.

  Consequently, we have $\lambda\not\models_{\MLT}\varphi_n$. But then
  the formula $\beta:=(\lambda\land\lnot\varphi_n)$ is satisfiable in
  $\MLT$ implying $S_\beta\models\beta$. From
  $\beta\models_{\MLK}\lambda$ and
  $\beta\models_{\MLK}\lnot\varphi_n$, we obtain
  $S_\beta\models\beta\land\lnot\varphi_n\land\lambda$. Hence
  $\beta\in B$. Now $S_\beta\models\lambda$ implies
  $(S,\iota_\beta)\models\lambda$ by \eqref{eq:substructure} and
  therefore $(S,\kappa)\models\nxt\lambda$.

  This finishes the proof of the claims \ref{claim-4-1}
  and \ref{claim-4-2}. By Lemma~\ref{L-E-and-O} and
  $(S,\kappa)\models\psi$, we obtain $(S,\iota)\models\psi$.
  
  \paragraph{Conclusion of Case 1}  Through
  steps~\ref{claim-1}~to~\ref{claim-4}, we proved
  $(S,\iota)\models\lnot\varphi_{n+1}\land\psi$ in case
  $\bigl|\bigvee O_\psi\bigr|<2^n$, contradicting the equivalence of
  $\varphi_{n+1}$ and $\psi$.

  \paragraph{Case 2,
    $\sum_{\lambda\in E_\psi} \lnxtcount{\lambda} < 2^n$}

  We consider the formula
  $\lnot p_{(n+1)\bmod2} \lor \bigl((\lnot p)\liff\nxt\varphi_n\bigr)
  \equiv_{\MLK} \lnot\varphi_{n+1} \equiv_{\MLT} \lnot\psi$.
  Observe that $O_{\lnot\psi}=E_\psi$, such that
  $\left|\bigvee O_{\lnot\psi}\right|<2^n$. Hence we can use the same
  argument as before for $\lnot\psi$ to obtain a contradiction: simply
  force $p$ to hold in the two worlds $\iota$ and $\kappa$. We provide
  the details in the appendix.

  Having derived contradictions in both case 1 and case 2, it follows
  that $\lnxtcount{\psi}<2^{n+1}$ is not possible. But this finishes 
  the inductive proof.
\end{proof}

From the above lemma, we immediately get that $\ML[\extdM]$ does not
have sub-exponential translations wrt.~$\MLT$ in $\ML[\dM]$. Now let
$G$ be any finite and complete set of Boolean functions such that at
least one function from $G$ is not locally monotone.
Then, by Theorem~\ref{T-main}, the logic $\ML[\extdM]$ has polynomial
translations wrt.~\MLT in $\ML[G]$. Since the relation ``has
polynomial translations wrt.~\MLT'' is transitive, the above lemma
ensures that $\ML[G]$ does not have polynomial translations wrt.~\MLT
in $\ML[\dM]$ (for otherwise $\ML[\extdM]$ would have such
translations).

The class of sub-exponential functions is not closed under composition
(e.g., the functions $n\mapsto n^2$ and $n\mapsto 2^{\sqrt n}$ are
sub-exponential, but their composition is $n\mapsto 2^n$). Hence we
cannot use the above argument to show that $\ML[G]$ does not have
sub-exponential translations wrt.~\MLT in $\ML[\dM]$.
To show this stronger result, we adopt the above proof as follows.

{
\newcommand\Fiff{\ensuremath{\mathtt{iff}}\xspace}%
\newcommand\Fand{\ensuremath{\mathtt{and}}\xspace}%
\newcommand\Fnon{\ensuremath{\mathtt{neg}}\xspace}%

Let $G^+=G\cup\{\lverum,\lfalsum\}$. From Propositions
\ref{P-generate-negation}, \ref{P-generate-disjunction}, and
\ref{P-generate-biimplication}, we obtain $\PL[G^+]$-formulae
$\Fiff(x,y)\equiv_{\MLK} x\liff y$,
$\Fand(x,y)\equiv_{\MLK} x\land y$, and
$\Fnon(y)\equiv_{\MLK} \lnot y$, each using the variable $y$ only
once. We can therefore express the $\ML[\extdM]$-formulae $\varphi_n$
in a natural way by equivalent $\ML[G^+]$-formulae $\tilde\varphi_n$:
\begin{align*}
  \tilde\varphi_0     &= \Fand\formulasubst[\Big]{\ 
                           p_0,\ 
                           \nxt
                           \Fnon\formulasubst{p_0}\ 
                         }
                         &&\text{and} &&&&&&\\
  \tilde\varphi_{n+1} &= \Fand\formulasubst[\Big]{\ 
                           p_{(n+1)\bmod2},\ 
                           \Fiff\formulasubst[\big]{
                             p,
                             \nxt\varphi_n
                           }\ 
                         }
                         &&\text{for }n\geq 0\,.
\end{align*}
By choice of \Fiff and \Fand, the formula $\nxt\varphi_n$ occurs only
once in $\tilde\varphi_{n+1}$. Since the formulae $\Fiff$, $\Fand$,
and $\Fnon$ remain fixed, the size of $\tilde\varphi_n$ is linear in
$n$, as is the size of $\varphi_n$.  Note that the formula
$\tilde\varphi_n$ can contain the formulae $\lverum,\lfalsum\in G^+$
which need not belong to $G$. Since $G$ is complete however, we can
express them by some fixed $\ML[G]$-formulae possibly using a new
propositional variable $z$. Let $\overline{\varphi_n}$ denote the
result of these replacements. Then the size of $\overline{\varphi_n}$
is linear in that of $\tilde\varphi_n$ and therefore in
$n$. Furthermore,
$\overline{\varphi_n} \equiv_{\MLK} \tilde\varphi_n \equiv_{\MLK}
\varphi_n$. Now the proof of Lemma~\ref{L-psi-large} works for
$\ML[G]$ and we obtain the main result of the section.  }

\begin{thm}\label{T:exp-gap-in-MLT/K4}
  Let $G$ be a finite and  complete set of Boolean
  functions with some function that is not locally monotone.
  Then the logic $\ML[G]$ does not have sub-exponential translations 
  wrt.~$\MLT$ in $\ML[\dM]$.
\end{thm}

\subsection{Just one succinctness class}

In this section we show that, wrt. the modal logic $\mathrm{S5}$,
there is only one succinctness class. Recall that $\MLSf$ is the class
of pointed Kripke structures whose accessibility is an equivalence
relation.

\begin{thm}\label{T:one-succinctness-class-in-S5}
  Let $F$ be a finite set of Boolean functions. Then $\ML[F]$ has
  polynomial translations wrt.~\MLSf in $\ML[\dM]$.
\end{thm}

In view of Theorem~\ref{T-main}, it suffices to show the claim for
$F=\extdM$, i.e., that $\ML[\extdM]$ has polynomial translations
wrt.~\MLSf in $\ML[\dM]$ (for more details, we refer to the proof of
Theorem~\ref{T:one-succinctness-class-in-S5} on
page~\pageref{T:one-succinctness-class-in-S5:proof}). Hence we will,
for the major part of the section, consider the (extended) De~Morgan
basis.

Let $\varphi\in\ML[\extdM]$. Recall that $|\varphi|$ denotes the
number of nodes in the syntax tree of~$\varphi$. In this section, we
use the following additional size measures of $\varphi$ and its syntax
tree:
\begin{itemize}
\item the \emph{norm} $\|\varphi\|$ of $\varphi$, denoting the number
  of leaves in the syntax tree, i.e., the total number of occurrences
  of propositional variables and constants $\lverum$ and $\lfalsum$,
  as well as
\item the \emph{depth} $d(\varphi)$ of $\varphi$, denoting the depth
  of the syntax tree, where $d(\varphi)=0$ if, and only if, the tree
  consists of a single leaf, i.e., if
  $\varphi\in\Var\cup\{\lverum,\lfalsum\}$.
\end{itemize}
Since the arity of all functions that occur in $\varphi$ is bounded by
$2$, the number of leaves, the size, and the depth of $\varphi$
satisfy $1\leq \|\varphi\| \leq |\varphi|\leq 2^{d(\varphi)+1}-1$.

We first show that, wrt.~$\MLSf$, formulae can be balanced, i.e., that 
any formula in $\ML[\extdM]$ can be rewritten so as to have logarithmic
depth in the norm of the initial formula.
A similar result is already known for $\PL$ 
(see Remark~\ref{R:balancing-in-PL}).

\begin{lem}\label{L:balancing-in-MLSf}
  For every $\varphi\in\ML[\extdM]$ there exists a formula
  $\varphi'\in\ML[\extdM]$ with $\varphi'\equiv_{\MLSf}\varphi$ and 
  $d(\varphi') \leq 8 \cdot (1+\log_2 \|\varphi\|)$.
\end{lem}

\begin{rem}\label{R:balancing-in-PL}
  Since $\PL[\extdM]\subseteq\ML[\extdM]$, the lemma implies that
  each $\PL[\extdM]$-formula $\varphi$ is equivalent to some
  $\PL[\extdM]$-formula of depth logarithmic in $\|\varphi\|$.
  According to Gashkov and Sergeev \cite{GasS20}, a more general form
  of this result for propositional logic was known to Khrapchenko in
  1967, namely that it holds for any complete basis of Boolean
  functions (e.g., for $\{\land,\lor,\lnot,\liff\}$ as here). They
  also express their regret that the only source for this is  a
  single paragraph in a survey article by Yablonskii and Kozyrev
  \cite{YabK68}, see also \cite{Juk12}. Often, it is referred to as
  Spira's theorem who published it in 1971 \cite{Spi71}, assuming that
  all at most binary Boolean functions are allowed in propositional
  formulae. Khrapchenko's general form was then published by Savage
  \cite{Sav76}.
\end{rem}

\begin{proof}
  Differently from other results and proofs in this paper, this lemma
  is only concerned with one logic, namely $\ML[\extdM]$. We will
  therefore simply speak of 'formulae' when actually referring to
  $\ML[\extdM]$-formulae.

  The proof proceeds by induction on the number of leaves
  $\|\varphi\|$ in the syntax tree of the formula $\varphi$.
  
  First, suppose that $\|\varphi\|=1$. Then there is an integer $r\geq 0$,
  functions $f_1,\ldots,f_r\in\{\nxt,\lnot\}$, and an atomic formula
  $\lambda\in \Var\cup\{ \lverum,\lfalsum \}$ such that
  $\varphi=f_1\formulabrackets{
             f_2\formulabrackets{
              \cdots
               f_r\formulabrackets{\lambda}\cdots}}$ 
  -- that is, $\varphi$ is a sequence of negations and modal
  operators, terminating with the atomic formula $\lambda$. It is well
  known that, in $\MLSf$, any such formula is equivalent to a sequence
  of at most three negations and/or modal operators that ends with 
  $\lambda$, i.e., to a formula of depth at most 
  $3\le 8\cdot(1+\log_2\|\varphi\|)$ (see, e.g., Section 4.4 in \cite{Che80}).
  This establishes the case $\|\varphi\|=1$.

  Otherwise, $\|\varphi\|\geq 2$ and $\varphi$ contains at least one of
  the binary functions $\land$, $\lor$, and $\liff$. Let $m=\|\varphi\|$. 
  Intuitively, we split the formula $\varphi$ into two parts, each 
  containing about half the leaves from $\varphi$. Formally, there are 
  formulae $\alpha=\alpha(x)$, with only one occurrence of $x$, and $\beta$
  such that 
  \begin{itemize}
    \item $\varphi=\alpha(\beta)$,
    \item $\|\alpha\| \le \frac{m}{2} < \|\beta\|$, and
    \item $\beta=f\formulabrackets{\beta_1,\beta_2}$ with
          $f\in\{\land,\lor,\liff\}$ and
          $\|\beta_1\|,\|\beta_2\|\leq\frac{m}2$.
  \end{itemize}
  It is not difficult to find such formulae $\alpha$ and $\beta$:
  simply start at the root of the syntax tree of $\varphi$ and, while
  possible, proceed towards the child that contains more than half the
  leaves of $\varphi$. The procedure stops in a node, corresponding
  to a sub-formula $\beta$ of $\varphi$, which contains more than
  $\frac{\|\varphi\|}2$ leaves, but for which each child contains less
  than $\frac{\|\varphi\|}2$ leaves. In particular, there are two
  children, hence $\beta=f\formulabrackets{\beta_1,\beta_2}$ for
  some $f\in\{\land,\lor, \liff\}$. Substituting this particular
  sub-formula $\beta$ in $\varphi$ by $x$ results in the formula
  $\alpha(x)$.
  
  Note that
  $\|\alpha\| = m - \|\beta\| + 1 
              < m - \frac{m}2 + 1 = \frac{m}{2}+1$, 
  i.e., $\|\alpha\|\leq \frac{m}2$.

  First assume that $x$ does not occur within the scope of a
  $\nxt$-operator in $\alpha(x)$. Then $\beta$ is interpreted in the
  initial world implying
  \begin{align}\label{eq:MLSf-pull-phi-apart}
    \varphi=\alpha(\beta) 
    \equiv_{\MLSf} 
        \formulabrackets[\big]{\alpha(\lfalsum)\land\lnot\beta}
    \lor\formulabrackets[\big]{\alpha(\lverum)\land\beta}\,.
  \end{align}
  By induction hypothesis, there exist formulae 
  $\alpha'(x)\equiv_{\MLSf}\alpha(x)$, 
  $\beta_1'\equiv_{\MLSf}\beta_1$, and 
  $\beta_2'\equiv_{\MLSf}\beta_2$ of depth 
  $\leq 8\cdot (1+\log_2 \frac{m}{2}) = 8 \cdot \log_2 m$.
  Set $\beta'=f\formulabrackets{\beta_1',\beta_2'}$. Then $\beta'$
  is equivalent to $\beta$ over~\MLSf and of depth 
  $\le 1+8\cdot\log_2 m$. 
  Together with~\eqref{eq:MLSf-pull-phi-apart}, it follows that 
  \[
    \varphi
      \equiv_{\MLSf}
    \formulabrackets[\big]{\alpha'(\lfalsum) \land \lnot\beta'} 
     \lor
    \formulabrackets[\big]{\alpha'(\lverum)  \land \beta'}\,.
  \]
  Let $\varphi'$ denote the formula on the right hand side of this
  equivalence. Then the depth of $\varphi'$ satisfies
  \begin{align*}
    d(\varphi') &= \max\left\{ 2+d\bigl(\alpha'\bigr),\ 
                               3+d\bigl(\beta'\bigr) \right\} \\
                &\le \max\left\{ 2+8\cdot\log_2 m,\ 
                               3+1+8\cdot\log_2 m\right\}
                 \leq 4 + 8\cdot\log_2 m \\
                &\leq 8\cdot(1+\log_2 m)\,.
  \end{align*}
  This completes the first case, where $x$ does not occur within the
  scope of a $\nxt$-operator.

  Otherwise, we split $\alpha(x)$ at the last $\nxt$-operator above
  $x$, i.e., there are formulae $\alpha_1(y)$ with only one occurrence
  of $y$ and $\alpha_2(x)$ with only one occurrence of $x$ that,
  furthermore, does not lie in the scope of a $\nxt$-operator, such
  that $\alpha(x)=\alpha_1(\nxt\alpha_2(x))$. In particular,
  $\|\alpha_1\|,\|\alpha_2\|\leq\|\alpha\|\leq\frac{m}2$.  Hence, by
  induction hypothesis, there are formulae
  $\alpha_1'(y),\alpha_2'(x), \beta_1'$, and $\beta_2'$ with
  \begin{itemize}
    \item $\alpha_i'(z)\equiv_{\MLSf}\alpha_i(z)$ and 
          $d(\alpha_i')\leq 8\cdot\log_2 m$ for $i\in\{1,2\}$,
          as well as
    \item $\beta_i'\equiv_{\MLSf}\beta_i$ and 
          $d(\beta_i') \leq 8\cdot\log_2 m$ for $i\in\{1,2\}$.
  \end{itemize}
  As before, let $\beta'=f\formulabrackets{\beta_1',\beta_2'}$ with 
  $\beta'\equiv_{\MLSf}\beta$ and $d(\beta')\leq 1+8\cdot\log_2 m$. 

  Recall that, for a structure $K\in\MLSf$, the accessibility relation
  $R$ is
  an equivalence relation. Hence, for any two worlds $v$ and $w$ within
  the same equivalence class and any formula $\chi$, we have
  \begin{equation}\label{eq:MLSf-nxt-universal}
    K,v\models\nxt\chi \iff K,w\models\nxt\chi\,,
  \end{equation}
  i.e., whether or not $\nxt\chi$ holds in $K$ does not depend on the
  choice of the world (when considering only the relevant equivalence
  class containing the initial world).
  
  In $\alpha_2(\beta)$, the sub-formula $\beta$ is interpreted in the
  same world that the whole formula is interpreted in. With
  \[
    \psi'    = \formulabrackets[\big]{\alpha_2'(\lfalsum)\land\lnot\beta'}
                \lor
                \formulabrackets[\big]{\alpha_2'(\lverum) \land\beta'}
             \equiv_{\MLSf}
                \formulabrackets[\big]{\alpha_2(\lfalsum)\land\lnot\beta}
                \lor
                \formulabrackets[\big]{\alpha_2(\lverum) \land\beta}
              \]
   we therefore get $\alpha_2(\beta)\equiv_{\MLSf} \psi'$.
   Setting
   \begin{align*}
    \varphi' &= \formulabrackets[\big]{\alpha_1'(\lfalsum)\land\lnot\nxt\psi'}
                \lor 
                \formulabrackets[\big]{\alpha_1'(\lverum) \land\nxt\psi'} \,,
    \intertext{we furthermore get}
     \varphi'&\equiv_{\MLSf}
                \formulabrackets[\big]{\alpha_1(\lfalsum)\land\lnot\nxt\psi'}
                \lor
                \formulabrackets[\big]{\alpha_1(\lverum)\land\nxt\psi'} \\
             &\equiv_{\MLSf} 
                \alpha_1\bigl(\nxt\psi'\bigr) 
                && \text{by \eqref{eq:MLSf-nxt-universal}} \\
             &\equiv_{\MLSf}
                \alpha_1\bigl(\nxt\alpha_2(\beta)\bigr) \\
             &= \varphi\,.
  \end{align*}
  Then
  \begin{align*}
    d(\psi')    &= \max\left\{  2+d\bigl(\alpha_2'(x)\bigr),\ 
                                3+d\bigl(\beta'\bigr) \right\}
              \leq 4+8\cdot\log_2m \quad\text{ and} \\
    d(\varphi') &= \max\left\{  2+d\bigl(\alpha_1'(y)\bigr),\ 
                                4+d\bigl(\psi'\bigr) \right\} \\
             &\leq \max\left\{  2+8\cdot\log_2m,\ 
                                8+8\cdot\log_2m \right\}
              \leq 8\cdot(1+\log_2 m)\,,
  \end{align*}
  which completes the second case, where $x$ occurs in the scope of a
  $\nxt$-operator, and hence the inductive proof.
\end{proof}

Let $\varphi\in\ML[\extdM]\setminus\ML[\dM]$. Since $\varphi$ contains
at least one occurrence of bi-implication, we obtain
$\|\varphi\|\ge2$. By the previous lemma, there exists an
$\ML[\extdM]$-formula $\psi$ that is equivalent to $\varphi$ over
\MLSf and of depth logarithmic in the norm of $\varphi$.  Let
$\varphi'$ be obtained from $\psi$ by replacing each sub-formula
$\alpha\liff\beta$ by
$\formulabrackets{\alpha\land\beta} \lor
\formulabrackets{\lnot\alpha\land\lnot\beta}$.  Then
$\varphi'\in\ML[\dM]$ and $\varphi'\equiv_{\MLSf}\varphi$.
Furthermore, the depth of $\varphi'$ increases at most by a factor of
three and therefore remains logarithmic in $\|\varphi\|$; more
precisely, $d(\varphi')\leq 3\cdot d(\psi) \leq 3\cdot 8\cdot
(1+\log_2\|\varphi\|) = 24 \cdot 2 \cdot \log_2\|\varphi\|$ since
$\|\varphi\|\ge2$.  Since all functions in $\extdM$ are at most
binary, it follows that
  \[ |\varphi'| \leq 2^{d(\varphi')+1}
                \leq 2^{48\cdot\log_2\|\varphi\|+1} 
                \leq c\cdot \|\varphi\|^{d}
                \leq c\cdot |\varphi|^{d}
     \quad\text{for some constants $c,d>0$.}
  \]
  Hence, we verified the following claim.
  
\begin{lem}\label{L:ML[iff]-translations-in-MLSf}
  $\ML[\extdM]$ has polynomial translations wrt.~\MLSf in $\ML[\dM]$.
\end{lem}

Therefore, the modal logic $\mathrm{S5}$ has only one ``succinctness
class'', as stated in Theorem~\ref{T:one-succinctness-class-in-S5}.

\begin{proof}[Proof of
  Theorem~\ref{T:one-succinctness-class-in-S5}]\label{T:one-succinctness-class-in-S5:proof}
  Let $F$ be a finite set of Boolean functions. 
  Since bi-implication $\liff$ is not locally monotone, it follows by 
  Theorem~\ref{T-main} that $\ML[F]$ has polynomial
  translations (wrt.~\MLK) in $\ML[\extdM]$ and hence, particularly, also
  wrt.~\MLSf. By Lemma~\ref{L:ML[iff]-translations-in-MLSf} above,
  $\ML[\extdM]$ again has polynomial translations wrt.~\MLSf in
  $\ML[\dM]$. Since the relation ''has polynomial translations'' is
  transitive, this establishes the claim.
\end{proof}

\section{Conclusion}

This paper considers the impact of the choice of a complete Boolean
basis on the succinctness of logical formulae. It is show that the
succinctness depends solely on the existence of some non-locally
monotone operator in the basis: if there is no such function, then the
logical formulae using this basis agree in succinctness (up to a
polynomial) with those using the De~Morgan basis, otherwise, they
agree with the formulae using the extension of the De~Morgan basis with
bi-implication.

While this result is demonstrated for modal logic, the proof carries
over to many other classical logics like first-order logic, temporal
logic etc.

Regarding propositional logic, it was known before that all Boolean
bases give rise to the same succinctness. We show the same for the
modal logic $\mathrm{S5}$, i.e., when we restrict to Kripke structures
whose accessibility is an equivalence relation.  When considering all
reflexive Kripke structures (i.e., the modal logic $\mathrm{T}$)
however, this is no longer the case since then, the extended De~Morgan
basis allows formulae to be written exponentially more succinct than
the plain De~Morgan basis does.

It remains open, where exactly  this dichotomy occurs, namely
whether it holds when considering reflexive and transitive or
reflexive and symmetric accessibility relations. Furthermore, it is
not known whether the dichotomy holds for other logics like
first-order or temporal logic.

\newcommand{\etalchar}[1]{$^{#1}$}

%
\appendix
\section{Appendix}

In the induction step of the proof of Lemma~\ref{L-psi-large}, we only
spelled out one case and said that the second one is handled
similarly. Here, we provide the details of this missing case.

  \paragraph{Case 2,
    $\sum_{\lambda\in E_\psi} \lnxtcount{\lambda} < 2^n$}
  Let $E^\ast_\psi$ be the set of formulae $\lambda$ from $E_\psi$
  with $\lambda\models_{\MLT}\varphi_n$. Since
  $\lnxtcount{\bigvee E^\ast_\psi} \le \sum_{\lambda\in E_\psi}
  \lnxtcount{\lambda} < 2^n$, the induction hypothesis ensures
  $\bigvee E^\ast_\psi \not\equiv_{\MLT} \varphi_n$.  On the other
  hand, $\bigvee E^\ast_\psi \models_{\MLT} \varphi_n$ by choice of
  $E^\ast_\psi$. Hence the formula
  $\alpha=\varphi_n\land\lnot\bigvee E^\ast_\psi$ is satisfiable in
  $\MLT$ implying $S_\alpha\models\alpha$, i.e.,
  $S_\alpha\models\varphi_n$ but
  $S_\alpha\not\models\bigvee E^\ast_\psi$.
  
  Let $B$ denote the set of formulae $\beta$ with
  $S_\beta\models\beta\land\lnot\varphi_n$.
  
  We now define a Kripke structure $S=(W,R,V)$ as follows
  (cf.~Fig.~\ref{fig:trees-case2}):
  \begin{align*}
    W &= \{\iota,\kappa\}\uplus\bigcup_{\beta\in B}W_\beta \cup W_\alpha\\
    R &= \bigl\{(\iota,\iota),(\kappa,\kappa),
                (\iota,\kappa),(\iota,\iota_\alpha)\bigr\}
        \cup \Bigl(\{\iota,\kappa\}\times\{\iota_\beta \mid \beta\in B\}\Bigr)
        \cup \bigcup_{\beta\in B}R_\beta
        \cup R_\alpha\\
    V(q) &=\begin{cases}
       \bigcup_{\beta\in B} V_\beta(q) \cup V_\alpha(q) \cup \{\iota,\kappa\} 
      & \text{if } q\in\{p,p_{(n+1)\bmod2}\}\\
      \bigcup_{\beta\in B} V_\beta(q) \cup V_\alpha(q) 
      & \text{otherwise }
             \end{cases}
  \end{align*}

  Since the accessibility relations of the structures $S_\alpha$ and
  $S_\beta$ for $\beta\in H$ are reflexive, the same applies to the
  accessibility relation of $S$, i.e., we obtain
  $S\in\MLT$. Furthermore, 
  \begin{equation}
    \label{eq:substructure-case2}
    \text{for all }\lambda\in\ML[\dM], \gamma\in B\cup\{\alpha\}, w\in W_\gamma\colon \Bigl((S_\gamma,w)\models\lambda \iff (S,w)\models\lambda\Bigr)
  \end{equation}
  since we only add edges originating in $\iota$ or $\kappa$.

  \begin{figure}
    \centering
    \begin{tikzpicture}[structure]
      \node[world,
            label={[label distance=-2pt]135:
                   {$\substack{\iota\\ p, p_{(n+1)\bmod2}}$}}]  (i) {};
      \node[world,below right=.6cm and 3.1cm of i,
            label={[label distance=-2pt]0:
                   {$\substack{\kappa\\ p, p_{(n+1)\bmod2}}$}}] (i'){};

      \node[tree1,below left=2.3cm and -0cm of i]    (b1t) {$S_{\beta_1}$};
      \node[tree2,below left=2.3cm and -1.5cm of i]  (b2t) {$S_{\beta_2}$};
      \node[      below left=2.0cm and -2.8cm of i]        {$\dots$};
      \node[tree3,below left=2.3cm and -3.7cm of i]  (brt) {$S_{\beta_r}$};
      \node[      below left=2.0cm and -5cm of i]          {$\dots$};
      \node[tree1,below left=1.3cm and 1.6cm of i]   (at)  {$S_\alpha$};
      
      \node[pin,above=-5pt of b1t] (b1) {};
      \node[pin,above=-5pt of b2t] (b2) {};
      \node[pin,above=-5pt of brt] (br) {};
      \node[pin,above=-5pt of at]  (a) {};
      
      \draw (i)  edge[out= 20,in= 70,looseness=14] (i)
            (i') edge[out= 65,in=115,looseness=14] (i')
            (i) edge (i')
            (i)  edge (b1)
            (i)  edge (b2)
            (i)  edge (br)
            (i)  edge (a) 
            (i') edge[bend right=5] (b1)
            (i') edge (b2)
            (i') edge (br);
    \end{tikzpicture}
    \caption{Schematic representation of the Kripke structure $S$.}
    \label{fig:trees-case2}
  \end{figure}

  We will proceed by verifying the following claims:
  \begin{enumerate}[(A)]
  \item\label{claim-1-case2} $(S,\iota) \models \varphi_{n+1}$,
  \item\label{claim-2-case2} $(S,\kappa) \not\models \varphi_{n+1}$,
  \item\label{claim-3-case2} $(S,\kappa) \not\models \psi$, and
  \item\label{claim-4-case2} $(S,\iota) \not\models \psi$,
  \end{enumerate}
  thus showing $(S,\iota)\models\varphi_{n+1}\land\lnot\psi$, which
  contradicts the equivalence of $\varphi_{n+1}$ and $\psi$.

  \paragraph{Proof of~\ref{claim-1-case2},
    $(S,\iota)\models\varphi_{n+1}$}

  Recall that $(S_\alpha,\iota_\alpha)\models\varphi_n$ implying
  $(S,\iota_\alpha)\models\varphi_n$ by
  \eqref{eq:substructure-case2}. Hence we have
  $(S,\iota)\models p\land\nxt\varphi_n$ which ensures
  $(S,\iota)\models\varphi_{n+1}$.

  \paragraph{Proof of~\ref{claim-2-case2},
    $(S,\kappa)\not\models\varphi_{n+1}$}

  Towards a contradiction, suppose
  $(S,\kappa)\models\nxt\varphi_n$. Then there exists $w\in W$ with
  $(S,w)\models\varphi_n$ and $(\kappa,w)\in R$, i.e.,
  $w\in\{\kappa\}\cup\{\iota_\beta\mid \beta\in B\}$.  From
  $\varphi_n\models_{\MLK} p_{n\bmod2}$ and
  $(S,\kappa)\not\models p_{n\bmod2}$, we get
  $(S,\kappa)\models\lnot\varphi_n$ and therefore $w\neq\kappa$. Hence
  there is $\beta\in B$ with $w=\iota_\beta$. Now
  $(S,w)\models\varphi_n$ implies $(S_\beta,w)\models\varphi_n$. But
  this contradicts $\beta\in B$. Thus, indeed,
  $(S,\kappa)\not\models\nxt\varphi_n$.  Together with the observation
  $(S,\kappa)\models p_{(n+1)\bmod2}\land p$, we get
  $(S,\kappa)\not\models\varphi_{n+1}$.

  \paragraph{Claim~\ref{claim-3-case2}, $(S,\kappa)\not\models\psi$}
  This is immediate by the above since $S\in\MLT$ and
  $\psi \equiv_{\MLT} \varphi_{n+1}$.

  \paragraph{Claim~\ref{claim-4-case2}, $(S,\iota)\not\models\psi$}
  Towards a contradiction, assume $(S,\iota)\models\psi$.

  Recall that $\psi$ is a Boolean combination of atomic formulae and
  of formulae $\nxt\lambda$ with $\lambda\in O_\psi\cup E_\psi$. We
  now prove the following two claims:
  \begin{enumerate}[(D1)]
  \item\label{claim-4-1-case2} If $\lambda\in E_\psi$ with
    $(S,\iota)\models\nxt\lambda$, then
    $(S,\kappa)\models\nxt\lambda$.
  \item\label{claim-4-2-case2} If $\lambda\in O_\psi$ with
    $(S,\kappa)\models\nxt\lambda$, then
    $(S,\iota)\models\nxt\lambda$.
  \end{enumerate}

  \paragraph{Proof of~\ref{claim-4-1-case2}}
  Let $\lambda\in E_\psi$ with $(S,\iota)\models\nxt\lambda$.

  First, assume $\lambda\models_{\MLT}\varphi_n$, i.e.,
  $\lambda\in E^\ast_\psi$.  From $(S,\iota)\models\nxt\lambda$, we
  obtain that the formula $\lambda$ holds in one of the worlds
  $\iota$, $\kappa$, $\iota_\alpha$, or $\iota_\beta$ for some
  $\beta\in B$. But $(S,\iota_\alpha)\models\lambda$ implies
  $S_\alpha\models\lambda$ by \eqref{eq:substructure-case2}, which is
  impossible since $S_\alpha\models\alpha$,
  $\alpha\models_{\MLK} \lnot\bigvee E^\ast_\psi$, and
  $\lnot\bigvee E^\ast_\psi\models_{\MLK} \lnot\lambda$ since
  $\lambda\in E^\ast_\psi$. Next suppose
  $(S,\iota_\beta)\models\lambda$ for some $\beta\in B$. Then
  $S_\beta\models\lambda$ which, together with
  $\lambda\models_{\MLT}\varphi_n$, implies
  $S_\beta\models\varphi_n$. Because of $\beta\in B$, we have
  $S_\beta\models\lnot\varphi_n$, a contradiction.

  Consequently, the formula $\lambda$ holds in one of the worlds
  $\iota$ and $\kappa$. Again using $\lambda\models_{\MLT}\varphi_n$,
  we obtain that also $\varphi_n$ holds in $\iota$ or in $\kappa$. But
  this cannot be the case since $p_{n\bmod 2}$ does not hold in either
  of the two worlds -- a contradiction. 

  Hence, we have $\lambda\not\models_{\MLT}\varphi_n$. But then the
  formula $\beta:=(\lambda\land\lnot\varphi_n)$ is satisfiable in
  $\MLT$ implying $S_\beta\models\beta$. From
  $\beta\models_{\MLK}\lambda$ and
  $\beta\models_{\MLK}\lnot\varphi_n$, we obtain
  $S_\beta\models\beta\land\lnot\varphi_n\land\lambda$. Hence
  $\beta\in B$. Now $S_\beta\models\lambda$ implies
  $(S,\iota_\beta)\models\lambda$ and therefore
  $(S,\kappa)\models\nxt\lambda$.

  \paragraph{Proof of~\ref{claim-4-2-case2}}
  Let $\lambda\in O_\psi$ with $(S,\kappa)\models\nxt\lambda$.  Then
  there is $w\in\{\kappa\}\cup\{\iota_\beta\mid\beta\in B\}$ such that
  $(S,w)\models\lambda$. In any case, $(\iota,w)\in R$. Hence we
  have $(S,\iota)\models\nxt\lambda$.

  This finishes the proof of the claims \ref{claim-4-1-case2} and
  \ref{claim-4-2-case2}.  Note that $(S,\iota)$ and
  $(S,\kappa)$ agree in the atomic formulae holding there. From
  Lemma~\ref{L-E-and-O} and our assumption $(S,\iota)\models\psi$, we
  therefore get $(S,\kappa)\models\psi$, contrary to what we saw
  before. Hence, indeed, $(S,\kappa)\not\models\psi$.
  
  \paragraph{Conclusion of Case 2}  Through
  steps~\ref{claim-1-case2}~to~\ref{claim-4-case2}, we proved
  $(S,\iota)\models\varphi_{n+1}\land\lnot\psi$ also in case
  $\bigl|\bigvee E_\psi\bigr|<2^n$, contradicting the equivalence of
  $\varphi_{n+1}$ and $\psi$.

\end{document}